\documentclass[a4paper,11pt]{article}

\usepackage{graphicx} 	
\usepackage[tight]{subfigure} 
\usepackage{float}

\graphicspath{{./Images/}}

\usepackage[ pdftex, plainpages = false, pdfpagelabels, 
           bookmarks=false,
           bookmarksopen = true,
           bookmarksnumbered = true,
           breaklinks = true,
           linktocpage,
           pagebackref,
           colorlinks = true, 
           linkcolor = blue,
           urlcolor = blue,
           citecolor = red,
           anchorcolor = green,
           hyperindex = true,
           hyperfigures
           ]{hyperref} 

\usepackage{amsthm}
\usepackage{amsmath}
\usepackage{amssymb}
\usepackage{amsfonts}
\usepackage{amstext}
\usepackage{xspace}
\usepackage{graphicx}
\usepackage{mathrsfs}
\usepackage[shortlabels]{enumitem}
\setlist{noitemsep}

\usepackage{algorithm}
\usepackage[noend]{algorithmic}
\usepackage{algorithmic-fix}
\renewcommand{\algorithmicrequire}{\textbf{Invariant:}}
\renewcommand{\algorithmicensure}{\textbf{Postcondition:}}

\usepackage{vmargin}
\setmarginsrb{1in}{1in}{1in}{1in}{0pt}{0pt}{0pt}{6mm}

\newcommand{\yes}{\textsc{yes}\xspace}
\newcommand{\no}{\textsc{no}\xspace}
\newcommand{\yesinstance}{\yes-instance\xspace}
\newcommand{\noinstance}{\no-instance\xspace}

\newtheorem*{untheorem}{Theorem}
\newtheorem{theorem}{Theorem}
\newtheorem{lemma}{Lemma}[section]
\newtheorem{claim}{Claim}[section]

\newtheorem{definition}{Definition}[section]
\newtheorem{observation}{Observation}[section]
\newtheorem{proposition}{Proposition}[section]

\newtheorem{define}{Definition}

\newenvironment{claimproof}{\begin{proof}\renewcommand{\qedsymbol}{\claimqed}}{\end{proof}\renewcommand{\qedsymbol}{\plainqed}}

\let\plainqed\qedsymbol

\newcommand{\image}[1]{\includegraphics[scale=0.35]{#1}}
\newcommand{\fpt}{{\sf FPT}}

\newcommand{\Oh}{{\mathcal{O}}}

\newcommand{\Q}{{\mathcal{Q}}}

\newcommand{\td}{\mathbf{td}}

\newcommand{\height}{{\mathbf{height}}}
\newcommand{\depth}{{\mathbf{depth}}}
\newcommand{\reach}{{\mathbf{reach}}}
\newcommand{\anc}{{\mathbf{anc}}}

\newcommand{\GG}{{\cal G}}

\newcommand{\tw}{{\mathbf{tw}}}

\newcommand{\containment}[0]{{\sf NP}~$\subseteq$~{\sf coNP$/$poly}\xspace}

\newcommand{\h}[1]{\end{document}}

\DeclareMathOperator{\operatorClassNP}{\sf NP}

\usepackage{boxedminipage}

\usepackage{todonotes}

\newcommand{\X}{\ensuremath{\mathcal{X}}\xspace}
\newcommand{\C}{\ensuremath{\mathcal{C}}\xspace}
\newcommand{\F}{\ensuremath{\mathcal{F}}\xspace}
\newcommand{\T}{\ensuremath{\mathcal{T}}\xspace}

\newcommand{\poly}{{\rm{poly}}}

\newcommand{\cO}{{\mathcal{O}}}

\newcommand{\defparproblem}[4]{
 \vspace{1mm}
\noindent\fbox{
 \begin{minipage}{0.96\textwidth}
 \begin{tabular*}{\textwidth}{@{\extracolsep{\fill}}lr} #1 & {\bf{Parameter:}} #3 \\ \end{tabular*}
 {\bf{Input:}} #2 \\
 {\bf{Question:}} #4
 \end{minipage}
 }
 \vspace{1mm}
}

\usepackage{bold-extra}
\usepackage{lmodern}
\usepackage[T1]{fontenc}
\usepackage{textcomp}

\newcommand{\ProblemFormat}[1]{{\sc #1}}

\newcommand{\ProblemIndex}[1]{\index{problem!\ProblemFormat{#1}}}

\newcommand{\ProblemName}[1]{\ProblemFormat{#1}\ProblemIndex{#1}\xspace}
\newcommand{\ProblemNameX}[2]{\ProblemFormat{#1}\ProblemIndex{#2}\xspace}

\newcommand{\probtdeta}{\ProblemName{Treedepth-$\eta$ Deletion}}
\newcommand{\XSC}{\ProblemName{Exact $d$-Uniform Set Cover}}
\newcommand{\FDeletion}{\ProblemName{$\F$-Minor-Free Deletion}}
\newcommand{\PlanarFDeletion}{\ProblemName{Planar $\F$-Minor-Free Deletion}}
\newcommand{\TreewidthEtaDeletion}{\ProblemName{Treewidth-$\eta$ Deletion}}
\newcommand{\TreedepthEtaDeletion}{\probtdeta}
\newcommand{\KdplusoneDeletion}{\ProblemName{$\{K_{d+1}\}$-Minor-Free Deletion}}
\newcommand{\CliqueorPathDeletion}{\ProblemName{$\{K_{d+1},P_{4d}\}$-Minor-Free Deletion}}
\newcommand{\TreewidthDminoneDeletion}{\ProblemName{Treewidth-$(d-1)$ Deletion}}

\title{Uniform Kernelization Complexity of Hitting Forbidden Minors}
\author{Archontia C. Giannopoulou\thanks{School of Engineering and Computing Sciences, Durham University, United Kingdom. \texttt{Archontia.Giannopoulou@gmail.com}. 
Supported by ERC Grant Agreement nr.~267959.}
\and Bart M. P. Jansen\thanks{Department of Mathematics and Computer Science, Technical University Eindhoven, The Netherlands. \texttt{b.m.p.jansen@tue.nl}. Supported by NWO Veni grant ``Frontiers in Parameterized Preprocessing'' and NWO Gravity grant ``Networks''.}  
\and Daniel Lokshtanov\thanks{Department of Informatics, University of Bergen, Norway. \texttt{daniello@ii.uib.no}. Supported by Bergen Research Foundation grant BeHard.}
\and Saket Saurabh\thanks{The Institute of Mathematical Sciences, Chennai, India. \texttt{saket@imsc.res.in}. Supported by Parameterized Approximation, ERC Starting Grant 306992.}
}

\date{}
\begin{document}
\maketitle

\thispagestyle{empty}
\vspace{-.7cm}
\begin{abstract}
The \FDeletion problem asks, for a fixed set~\F and an input consisting of a graph~$G$ and integer~$k$, whether~$k$ vertices can be removed from~$G$ such that the resulting graph does not contain any member of~\F as a minor. It generalizes classic graph problems such as \textsc{Vertex Cover} and \textsc{Feedback Vertex Set}. This paper analyzes to what extent provably effective and efficient preprocessing is possible for \FDeletion. Fomin et al.~(FOCS 2012) showed that the special case \PlanarFDeletion (when~\F contains at least one planar graph) has a kernel of polynomial size: instances~$(G,k)$ can efficiently be reduced to equivalent instances~$(G',k)$ of size~$f(\F) \cdot k^{g(\F)}$ for some functions~$f$ and~$g$. The degree~$g$ of the polynomial grows very quickly; it is not even known to be computable. Fomin et al.\ left open whether \PlanarFDeletion has kernels whose size is \emph{uniformly polynomial}, i.e., of the form~$f(\F) \cdot k^c$ for some universal constant~$c$ that does not depend on~$\F$. Our results in this paper are twofold.

\begin{enumerate}
	\item We prove that not all \PlanarFDeletion problems have uniformly polynomial kernels (unless \containment). Since a graph class has bounded treewidth if and only if it excludes a planar graph as a minor, a canonical \PlanarFDeletion problem is \TreewidthEtaDeletion: can~$k$ vertices be removed to obtain a graph of treewidth at most~$\eta$? We prove that the \TreewidthEtaDeletion problem does not have a kernel with~$\Oh(k^{\frac{\eta}{4} - \epsilon})$ vertices for any~$\epsilon > 0$, unless \containment. In fact, we prove the stronger result that even parameterized by the \emph{vertex cover number} of the graph (a larger parameter), the \TreewidthEtaDeletion problem does not admit uniformly polynomial kernels unless \containment. This resolves an open problem of Cygan et al.~(IPEC 2011). It is a natural question whether further restrictions on~$\F$ lead to uniformly polynomial kernels. However, we prove that even when~\F contains a \emph{path}, the degree of the polynomial must, in general, depend on the set~\F.
	\item Since a graph class has bounded treedepth if and only if it excludes a path as a minor, a canonical \FDeletion problem when~\F contains a path is \TreedepthEtaDeletion: can~$k$ vertices be removed to obtain a graph of treedepth at most~$\eta$? We prove that \TreedepthEtaDeletion admits uniformly polynomial kernels: for every fixed~$\eta$ there is a polynomial kernel with~$\Oh(k^6)$ vertices. In order to develop the kernelization we prove several new results about the structure of optimal treedepth-decompositions. These insights allow us to formulate a simple, fully explicit, algorithm to reduce the instance. As opposed to the kernels of Fomin et al.~(FOCS 2012), our kernelization algorithm does not rely on ``protrusion machinery'', which is a source of algorithmic non-constructivity.
\end{enumerate}
\end{abstract}

\newpage

\setcounter{page}{1}

\section{Introduction}

Kernelization is the subfield of parameterized and multivariate algorithmics that investigates the power of provably effective preprocessing procedures for hard combinatorial problems. While the origins of data reduction and preprocessing can be traced back far into the history of computing, the rigorous analysis of these topics developed over the last decade. In kernelization we study \emph{parameterized problems}: decision problems where every instance~$x$ is associated with a parameter~$k$ that measures some aspect of its structure. A parameterized problem is said to admit a kernel of size~$f \colon \mathbb{N} \to \mathbb{N}$ if every instance $(x,k)$ can be reduced in polynomial time to an equivalent instance with both size and parameter value bounded by~$f(k)$. For practical and theoretical reasons we are primarily interested in kernels whose size is polynomial, so-called \emph{polynomial kernels}. The study of kernelization has recently been one of the main areas of research in parameterized complexity, yielding many important new contributions to the theory. These include general results showing that  certain classes of parameterized problems have polynomial kernels, and results showing how to utilize algebra, matroid theory, and topology for data reduction~\cite{AlonGKSY11,H.Bodlaender:2009ng,FominLST10,Jansen14,Kratsch08Brief,KratschW12oct,KratschW12,PilipczukPSL13}. The development of a framework for ruling out polynomial kernels under certain complexity-theoretic assumptions~\cite{BDFH08,BodlaenderJK13a,DellM14,FortnowS11} has added a new dimension to the field and strengthened its connections to classical complexity.

One of the fundamental challenges in the area is the possibility of characterizing general classes of parameterized problems possessing a kernel of polynomial size. In other words, to obtain ``kernelization meta-theorems''. In general, algorithmic meta-theorems have the following form: problems definable in a certain logic admit a certain kind of algorithms on certain inputs.  A typical example of a meta-theorem is Courcelle's celebrated theorem~\cite{Courcelle92} which states that all graph properties definable in monadic second order logic can be decided in linear time on graphs of bounded treewidth. It seems very difficult to find a fragment of logic for which every problem expressible in this logic admits a polynomial kernel on all undirected graphs. The main obstacle in obtaining such results stems from the fact that even a simplest form of logic can formalize problems that are not even fixed parameter tractable (\fpt). In graph theory, one can define a general family of problems as follows. Let $\cal F$ be a family of graphs. Given an undirected graph $G$ and a positive integer $k$, is it possible to do at most $k$ edits of $G$ such that the resulting graph does not contain a graph from $\F$? Here one can define edits as either vertex/edge deletions, edge additions, or edge contraction. Similarly, one may consider containment as a subgraph, induced subgraph, or a minor. The topic of this paper is one such generic problem, namely, the \FDeletion problem.

The \FDeletion problem asks, for a fixed set of graphs~\F and an input consisting of a graph~$G$ and integer~$k$, whether~$k$ vertices can be removed from~$G$ such that the resulting graph does not contain any member of~\F as a minor. It generalizes classic graph problems such as \textsc{Vertex Cover}, \textsc{Feedback Vertex Set}, and \textsc{Vertex Planarization}. The parameterized complexity of this general problem is well understood. By a celebrated result of Robertson and Seymour, every  \FDeletion problem is non-uniformly  \fpt. That is, for every $k$ there is an algorithm solving the problem  in time~$f(k) \cdot n^3$~\cite{RobertsonS-GMXIII}.  However, whenever $\cal F$ is given explicitly, the problem is uniformly \fpt~because the excluded minors for the class of graphs that are \yesinstance of the \FDeletion~problem can by computed explicitly~\cite{AdlerGK08}. Thus, the \FDeletion problem is an interesting subject from the kernelization perspective:
\begin{quote}
For which sets~\F does \FDeletion admit a polynomial kernel?
\end{quote}

\noindent Fomin et al.~\cite{FominLMS12} studied the special case of \FDeletion problem where~$\F$ contains at least one planar graph, known as the \PlanarFDeletion problem. It is much more restricted than \FDeletion, but still generalizes problems such as \textsc{Vertex Cover} and \textsc{Feedback Vertex Set}. These problems are essentially about deleting~$k$ vertices to get a graph of constant treewidth: graphs that exclude a planar graph~$H$ as a minor have treewidth at most $|V(H)|^{\cO(1)}$~\cite{ChekuriC14}. In fact, a graph class has bounded treewidth if and only if it excludes a planar graph as a minor. Fomin et al.~\cite{FominLMS12} exploited the properties of graphs of bounded treewidth and obtained  a constant factor approximation algorithm, a~$2^{\cO(k \log k)} \cdot n$ time parameterized algorithm, and---most importantly, from our perspective---a polynomial sized kernel for every \PlanarFDeletion problem. More precisely, they showed that  \PlanarFDeletion admits a kernel of size $f(\F) \cdot k^{g(\F)}$ for some functions~$f$ and~$g$. The degree~$g$ of the polynomial in the kernel size grows very quickly; it is not even known to be computable. This result is the starting point of our research. 

\begin{quote}
Does \PlanarFDeletion have kernels whose size is \emph{uniformly polynomial}, of the form~$f(\F) \cdot k^c$ for a universal constant~$c$ that does not depend on~$\F$?
\end{quote}

\noindent We prove that some \PlanarFDeletion problems \emph{do not} have uniformly polynomial kernels (unless \containment). Since a graph class has bounded treewidth if and only if it excludes a planar graph as a minor, a canonical \PlanarFDeletion problem is \TreewidthEtaDeletion: can~$k$ vertices be removed to obtain a graph of treewidth at most~$\eta$? We denote by~$K_d$ and~$P_d$ a clique and path on~$d$ vertices, respectively. Our first theorem is the following lower bound result. 

\begin{theorem} \label{theorem:cliqueminordeletion:lowerbound}
Let~$d \geq 3$ be a fixed integer and~$\epsilon > 0$. If the parameterization by solution size~$k$ of one of the problems
\begin{enumerate}
	\item \KdplusoneDeletion, 
	\item \label{item:cliqueorpath} \CliqueorPathDeletion, and
	\item \TreewidthDminoneDeletion 
\end{enumerate}
admits a compression of bitsize~$\Oh(k^{\frac{d}{2} - \epsilon})$, or a kernel with $\Oh(k^{\frac{d}{4} - \epsilon})$ vertices, then \containment. In fact, even if the parameterization by the size~$x$ of a vertex cover of the input graph admits a compression of bitsize~$\Oh(x^{\frac{d}{2} - \epsilon})$ or a kernel with $\Oh(x^{\frac{d}{4} - \epsilon})$ vertices, then \containment.
\end{theorem}

\noindent 
Theorem~\ref{theorem:cliqueminordeletion:lowerbound} shows that the kernelization result of Fomin et al.~\cite{FominLMS12} is tight in the following sense: the degree~$g$ of the polynomial in the kernel sizes for \PlanarFDeletion must depend on the family \F. In fact, the theorem gives the stronger result that even parameterized by the \emph{vertex cover number} of the graph (a larger parameter), the \TreewidthEtaDeletion problem does not admit uniformly polynomial kernels unless \containment. This resolves an open problem of Cygan et al.~\cite{CyganLPPS12}.  As observed earlier, a graph class has bounded treewidth if and only if it excludes a planar graph as a minor. Thus, by restricting the \FDeletion problem to those \F that contain a planar graph, one exploits the properties of graphs of bounded treewidth to design polynomial kernels for \PlanarFDeletion. It is a natural question whether further restrictions on~$\F$ lead to uniformly polynomial kernels. However, the second item of Theorem~\ref{theorem:cliqueminordeletion:lowerbound} shows that even when~\F contains a \emph{path}, the degree of the polynomial must, in general, depend on the set~\F. This raises the question whether there are any general families of \FDeletion problems that admit uniformly polynomial kernels. 

Excluding planar minors results in graphs of bounded treewidth~\cite{RobertsonS-V}; excluding forest minors results in graphs of bounded pathwidth~\cite{RobertsonS83}; and excluding path minors results in graphs of bounded treedepth~\cite{NesetrilM06}. Since a graph class has bounded treedepth if and only if it excludes a path as a minor, a canonical \FDeletion problem when~\F contains a path is \TreedepthEtaDeletion.

\bigskip

\defparproblem{\probtdeta}{An undirected graph $G$ and a positive integer $k$.}{$k$}{Does there exist a subset $Z\subseteq V(G)$ of size at most~$k$ such that $\td(G-Z)\leq \eta$?}

\bigskip 

\noindent 
Here~$\td(G)$ denotes the treedepth of a graph~$G$. The set~$Z$ is called a \emph{treedepth-$\eta$ modulator} of~$G$. Surprisingly, we show that {\probtdeta} admits uniformly polynomial kernels. More precisely, we obtain the following theorem.

\begin{theorem}
\label{theorem:tdepthkernel}
\probtdeta admits a kernel with $2^{\cO(\eta^2)}k^6$ vertices.
\end{theorem}

We prove several new results about the structure of optimal treedepth decompositions and exploit this to obtain the desired kernel for \probtdeta. Unlike the kernelization algorithm of  Fomin et al.~\cite{FominLMS12}, our kernel is completely explicit. It does not use the machinery of protrusion replacement, which was introduced to the context of kernelization by Bodlaender et al.~\cite{H.Bodlaender:2009ng} and has subsequently been applied in various scenarios~\cite{FominLST10,FominLMPS11,GajarskyHOORRVS13,KimLPRRSS13}. Using protrusion replacement one can prove that kernelization algorithms exist, but the technique generally does not explicitly give the algorithm nor a concrete size bound for the resulting kernel.

\paragraph{Techniques.} The kernelization lower bound of Theorem~\ref{theorem:cliqueminordeletion:lowerbound} is obtained by reduction from \XSC, parameterized by the number of sets in the solution. Existing lower bounds exist for these problems due to Dell and Marx~\cite{DellM12} and Hermelin and Wu~\cite{HermelinW12}, showing that the degree of the kernel size must grow linearly with the cardinality~$d$ of the sets in the input. While the construction that proves Theorem~\ref{theorem:cliqueminordeletion:lowerbound} is relatively simple in hindsight, the fact that the construction applies to all three mentioned problems, and also applies to the parameterization by vertex cover number, makes it interesting.

Our main technical contribution lies in the kernelization algorithm for \TreedepthEtaDeletion. Our algorithm starts by enriching the graph~$G$ by adding edges between vertices that are connected by many internally vertex-disjoint paths. Like in prior work on \TreewidthEtaDeletion~\cite{CyganLPPS12}, adding such edges does not change the answer to the problem. We then apply an algorithm by Reidl et al.~\cite{ReidlRVS14} to compute an approximate treedepth-$\eta$ modulator~$S$ of the resulting graph. The remainder of the algorithm strongly exploits the structure of the bounded-treedepth graph~$G - S$. By combining separators for vertices that are not linked through many disjoint paths, we compute a small set~$Y$ such that all the bounded-treedepth connected components of~$G - (S \cup Y)$ have a special structure: their neighborhood in~$S$ forms a clique, while they have less than~$\eta$ neighbors in~$Y$. For such components~$C$ we can prove that optimal treedepth-$\eta$ modulators contain at most~$2\eta$ vertices from~$C$. This important fact allows us to infer that optimal solutions cannot disturb the structure of the graph~$G[C]$ too much. Using ideas inspired by earlier work on \textsc{Pathwidth}~\cite{BodlaenderJK12a}, it is relatively easy to bound the number of connected components of~$G - (S \cup Y)$. The main work consists of reducing the size of each such component.

We formulate three lemmata that analyze under which circumstances the structure of optimal treedepth-$\eta$ modulators is preserved when adding edges, removing edges, and removing vertices of the graph. By exploiting the fact that the solution size within a particular part~$C$ of the graph is constant, these lemmata ensure that even after deleting an optimal modulator from~$C$, the remainder of~$C$ forces a structure of treedepth decompositions of the remaining graph that is compatible with the graph modifications. Of particular interest is the lemma showing that if~$v$ dominates the neighborhood of component~$C$, then edges of~$v$ into the component may be safely discarded if certain other technical conditions are met.

The three described lemmata are the main tool in the reduction algorithm. To shrink components of~$G - (S \cup Y)$ we have to add some edges, while removing other edges, to create settings where vertices can be removed from the instance without changing its answer. The fact that we have to combine edge additions and removals makes our reduction algorithm quite delicate: we cannot simply formulate reduction rules for adding and removing edges and apply them exhaustively, as they would work against each other. We therefore present a recursive algorithm that processes a treedepth-$\eta$ decomposition of~$G - S$ from top to bottom, making suitable transformations that bound the degree of the modulator~$S$ into the remainder of the component~$C$. Using a careful measure expressed in terms of this degree, we can then prove that our algorithm achieves the desired size reduction.

\paragraph{Related Results.} \PlanarFDeletion has received considerable attention recently. To start with, Fomin et al.~\cite{FominLMS12} gave a $2^{\cO(k)} \cdot n$-time parameterized algorithm for a variant of \PlanarFDeletion where every graph in~$\cal F$ is connected. Kim et al.~\cite{KimLPRRSS13} showed that \PlanarFDeletion has an {\fpt} algorithm with running time $2^{\cO(k)} \cdot n^{2}$ time. They also showed, among many other results, that \PlanarFDeletion has linear kernel on topological-minor-free graphs. Cygan et al.~\cite{CyganLPPS12} studied the \TreewidthEtaDeletion problem parameterized by the vertex cover number of a graph and obtained a kernel of size~$k^{\cO(\eta)}$. In a later paper, Fomin et al.~\cite{FominJP14} studied \FDeletion parameterized by the vertex cover number of the graph. They obtained kernels of size~$k^{\cO(\Delta({\cal F}))}$, where~$\Delta({\cal F})$ is an upper bound on the maximum degree of any graph in $\cal F$.

Recently, treedepth has been the focus of several works. Reidl et al.~\cite{ReidlRVS14} gave an algorithm with running time $2^{\cO(t^2)} \cdot n$ to test whether the treedepth of graph is at most~$t$. Gajarsk\'{y} et al.~\cite{GajarskyHOORRVS13} obtained meta-theorems for kernelization when parameterized by a treedepth-$\eta$ modulator. They showed, for example, that problems satisfying certain technical conditions admit linear kernels on hereditary graphs of bounded expansion when parameterized by the size of a treedepth-$\eta$ modulator. 
\section{Preliminaries} \label{sec:prelim}
For a finite set~$X$ and non-negative integer~$n$ we use~$\binom{X}{n}$ to denote the collection of size-$n$ subsets of~$X$. We abbreviate~$\{1,\ldots,n\}$ by~$[n]$.

\subsection{Graphs}
All graphs we consider are finite, undirected, and simple. For a graph~$G$ we use~$V(G)$ to denote the vertex set and~$E(G)$ to denote the edge set, which is a subset of~$\binom{V(G)}{2}$. For graphs~$G$ and~$H$ we write~$H \subseteq G$ if~$H$ is a subgraph of~$G$, i.e., if~$V(H) \subseteq V(G)$ and~$E(H) \subseteq E(G)$. For~$S\subseteq V(G)$ we denote by $G - S$ the graph obtained from~$G$ after removing the vertices of~$S$ and their incident edges. In the case where~$S=\{u\}$, we abuse notation and write~$G - u$ instead of $G - \{u\}$. We denote by~$G[S]$ the subgraph of~$G$ induced by the set~$S$. For~$S\subseteq V(G)$, the \emph{open neighborhood} of~$S$ in~$G$, denoted~$N_{G}(S)$, is the set $\{u\in V(G) \setminus S \mid \exists v\in S \colon \{u,v\}\in E(G)\}$. Again, in the case where~$S=\{v\}$ we abuse notation and write~$N_{G}(v)$ instead of~$N_{G}(\{v\})$. The \emph{closed neighborhood} of~$S$ in~$G$, denoted~$N_G[S]$ is defined as~$N_G(S) \cup S$. Similarly,~$N_G[v] := N_G(v) \cup \{v\}$ for single vertices~$v$. The degree of a vertex~$v\in V(G)$, denoted by~$\deg_G(v)$, is~$\deg_{G}(v)=|N_G(v)|$. Given two distinct vertices~$u$ and~$v$ we define~$\lambda_{G}(u,v)$ as the maximum cardinality of a set of pairwise internally vertex-disjoint $uv$-paths in~$G$.

\subsubsection{Treedepth}
A \emph{rooted tree}~$T$ is a tree with one distinguished vertex~$r\in V(T)$, called the \emph{root} of~$T$. A \emph{rooted forest} is a disjoint union of rooted trees. The roots introduce natural parent-child and ancestor-descendant relations between vertices in forest. Let $x,y$ be vertices of a rooted forest~$F$. The vertex~$x$ is an \emph{ancestor} of~$y$ if~$x$ belongs to the path linking~$y$ to the root of the tree to which~$y$ belongs. It is a \emph{proper ancestor} if, in addition, it is not equal to~$y$. We denote by~$\anc_F(x)$ the proper ancestors of~$x$. Observe that this set is empty if~$x$ is the root of a tree. We may omit the index~$F$ if it is clear from the context. Vertex~$y$ is a \emph{descendant} of~$x$, if~$x$ is an ancestor of~$y$. A \emph{proper descendant} of~$x$ is a descendant that is not~$x$ itself. We denote by~$\pi(x)$ the parent of~$x$ in~$F$. The parent of the root of the tree is~$\bot$. Vertices whose parent is~$x$ are called the {\em{children}} of~$y$. 

For a rooted forest~$F$ and a vertex~$v\in V(F)$, we denote by~$F_v$ the subtree rooted at~$v$ that contains all~$v$'s descendants, including~$v$ itself. The \emph{depth} of a vertex~$x$ in a rooted forest~$F$ is the number of vertices on the unique simple path from~$x$ to the root of the tree to which~$x$ belongs. It is denoted by~$\depth(x,F)$. The \emph{height} of~$v$ is the maximum number of vertices on a simple path from~$v$ to a leaf in~$F_v$. The height of~$F$ is the maximum height of a vertex of~$F$ and is denoted by~$\height(F)$. Given a rooted forest~$F$ and a vertex~$v \in V(F)$ we define the \emph{reach} of~$v$ in~$F$ as
$$\reach(v,F) := \height(F)-\depth(v,F).$$ Intuitively, the reach of~$v$ shows the maximum height of a subtree that we can attach as a child of~$v$ without increasing the total height of the decomposition. Two vertices~$x$ and~$y$ are in \emph{ancestor-descendant} relation if~$x$ is an ancestor of~$y$ or vice versa.

\begin{define}[Treedepth]
A treedepth decomposition of a graph~$G$ is a rooted forest~$F$ on the vertex set~$V(G)$ (i.e.,~$V(G)=V(F)$) such that for every edge~$\{u,v\}$ of~$G$, the endpoints~$u$ and~$v$ are in ancestor-descendant relation. The \emph{treedepth} of~$G$, denoted~$\td(G)$, is the least $d \in\mathbb{N}$ such that there exists a treedepth decomposition~$F$ of~$G$ with~$\height(F)=d$.
\end{define}

We say that an edge~$\{p,q\}$ is represented in a treedepth decomposition if~$p$ and~$q$ are in ancestor-descendant relation. Observe that the treedepth of a disconnected graph is the maximum treedepth of its connected components.

\begin{observation}\label{obs:clque}
Let $G$ be a graph and $S\subseteq V(G)$ such that $G[S]$ is clique. If $F$ is a treedepth decomposition of $G$, then all the vertices of~$S$ belong to a root-to-leaf path of a tree~$T$ of~$F$.
\end{observation}

\begin{observation} \label{observation:connectedsubgraph:ancestorpath}
Let~$G$ be a graph with treedepth decomposition~$F$ and let~$H$ be a connected subgraph of~$G$. All vertices of~$G$ belong to the same tree~$T$ in~$F$. If~$u,v \in V(H)$ are not in ancestor-descendant relation in~$T$, then some vertex of~$H$ is a common ancestor of~$u$ and~$v$.
\end{observation}

\begin{observation} \label{observation:neighbors:of:subtree}
If~$F$ is a treedepth decomposition of~$G - S$ for some~$S \subseteq V(G)$ and~$v$ is a node in a tree~$T$ of~$F$, then all vertices of~$N_G(T_v)$ are ancestors of~$v$ or belong to~$N_G(T_v) \cap S$.
\end{observation}

We will work with the notion of a \emph{nice treedepth decomposition}. A treedepth decomposition~$F$ of a graph~$G$ is a nice treedepth decomposition if, for every~$v\in V(F)$, the subgraph of~$G$ induced by the vertices in~$F_v$ is connected.

\begin{lemma}[{\cite{ReidlRVS14}}] \label{lemma:opt:nice:decomposition}
For every fixed~$\eta$ there is a polynomial-time algorithm that, given a graph~$G$, either determines that~$\td(G) > \eta$ or computes a nice treedepth decomposition~$F$ of~$G$ of depth~$\td(G)$.
\end{lemma}

\begin{proposition} \label{proposition:nicedec:nothigher}
For any treedepth decomposition~$F$ of a graph~$G$, there exists a nice treedepth decomposition~$F^*$ of~$G$ whose height does not exceed the height of~$F$, such that no vertex has greater depth in~$F^*$ than in~$F$.
\end{proposition}
\begin{proof}
Let~$F$ be a treedepth decomposition of~$G$. While there is a node~$u \in F$ that is not a root and no vertex of~$F_u$ is adjacent in~$G$ to~$\pi(u)$, do the following. If~$\pi(u)$ is not a root, then remove the edge in~$F$ from~$u$ to~$\pi(u)$ and make~$u$ a child of~$\pi(\pi(u))$. If~$\pi(u)$ is a root, then remove the edge from~$u$ to~$\pi(u)$ and make~$u$ the root of the resulting tree in~$F$. It is easy to see that, since no vertex in~$F_u$ was adjacent to~$\pi(u)$, we do not violate the validity of the treedepth decomposition. It is also easy to see that the depth of vertices does not increase.

If the operation cannot be applied anymore, then for every~$u \in F$ that is not a root, the subtree~$F_u$ contains a vertex adjacent to~$\pi(u)$. A simple induction on the height of~$u$ then shows that~$G[F_u]$ is connected for every~$u$, implying that~$F$ is a nice treedepth decomposition.
\end{proof}

We use a known approximation algorithm for \probtdeta in the kernelization.

\begin{lemma}[{\cite[Lemma 2]{GajarskyHOORRVS13}}] \label{lemma:tdmodapprox}
Fix~$\eta \in \mathbb{N}$. Given a graph~$G$, one can in polynomial time compute a subset~$S\subseteq V(G)$ such that~$\td(G-S)\leq \eta$ and~$|S|$ is at most~$2^\eta$ times the size of a minimum treedepth-$\eta$ modulator of $G$.
\end{lemma}

\subsection{Parameterized complexity and kernelization} \label{section:lowerbound:tools}
In this section we define the concepts needed to prove lower bounds on kernelization. It will be convenient to consider lower bounds against compressions into small instances of arbitrary problems, rather than merely compressions of one parameterized problem into itself (which is a kernelization).

\begin{definition}[Compression] 
Let~$\Q,\Q' \subseteq \Sigma^* \times \mathbb{N}$ be parameterized problems and let~$f \colon \mathbb{N} \to \mathbb{N}$ be a function. A \emph{size-$f$ compression of~$\Q$ into~$\Q'$} is an algorithm that, given an instance~$(x,k) \in \Sigma^* \times \mathbb{N}$, takes time polynomial in~$|x| + k$ and outputs an instance~$(x', k') \in \Sigma^* \times \mathbb{N}$ such that:
\begin{enumerate}
	\item $(x,k) \in \Q$ if and only if~$(x',k') \in \Q'$, and
	\item both~$|x'|$ and~$k'$ are bounded by~$f(k)$.
\end{enumerate}
We say that~$\Q$ has a compression of size~$f$ if there is a parameterized problem~$\Q'$ for which there exists a size-$f$ compression of~$\Q$ into~$\Q'$. A \emph{kernelization of size~$f$} for a parameterized problem~$\Q$ is simply a compression of~$\Q$ into~$\Q$.
\end{definition}

To transfer lower bounds from one problem to another, we use the following type of reducibility.

\begin{definition}[Polynomial-parameter transformation]
Let~$\Q,\Q' \subseteq \Sigma^* \times \mathbb{N}$ be parameterized problems and let~$d \in \mathbb{N}$. A \emph{degree-$d$ polynomial-parameter transformation} from~$\Q$ to~$\Q'$ is an algorithm that, given an instance~$(x,k) \in \Sigma^* \times \mathbb{N}$, takes time polynomial in~$|x| + k$ and outputs an instance~$(x',k') \in \Sigma^* \times \mathbb{N}$ such that:
\begin{enumerate}
	\item $(x,k) \in \Q$ if and only if~$(x',k') \in \Q'$, and
	\item $k' \in \Oh(k^d)$.
\end{enumerate}
\end{definition}

Proposition~\ref{proposition:compressionbyppt} shows how to obtain a compression from a polynomial-parameter transformation.

\begin{proposition}\label{proposition:compressionbyppt}
Let~$\Q$ and~$\Q'$ be parameterized problems, and let~$c, d \in \mathbb{N}$. If there is a degree-$d$ polynomial-parameter transformation from~$\Q$ to~$\Q'$, and problem~$\Q'$ has a compression of size~$\Oh(k^c)$, then~$\Q$ has a compression of size $\Oh(k^{c \cdot d})$.
\end{proposition}
\begin{proof}
The compression algorithm for~$\Q$ works as follows. On input~$(x,k)$, it first applies the polynomial-parameter transformation to compute an equivalent instance~$(x',k')$ of~$\Q'$ whose parameter value~$k'$ is~$\Oh(k^d)$. It then applies the compression for~$\Q'$ to the instance~$(x',k')$, resulting in an instance~$(x^*,k^*)$. By the guarantee of the compression, the size and parameter of the compressed instance are bounded by~$\Oh((k^d)^c) \in \Oh(k^{c\cdot d})$. Since both steps preserve the answer and run in polynomial time, this is a valid compression for~$\Q$.
\end{proof}

For further background on parameterized complexity and kernelization we refer to one of the textbooks~\cite{DowneyF13,FlumG06,Niedermeier06} or recent surveys~\cite{Bodlaender09,Kratsch14,LokshtanovMS12}.
\section{Kernelization Lower Bounds}

We turn our attention to kernelization and compression lower bounds. To prove that \FDeletion does not have uniformly polynomial kernels for suitable families \F, we give a polynomial-parameter transformation from a problem for which a compression lower bound is known. The following problem is the starting point for our transformation.

\defparproblem{\XSC}
{A finite set~$U$ of size~$n$, an integer~$k$, and a set family~$\F \subseteq 2^{U}$ of size-$d$ subsets of~$U$.}
{The universe size~$n$.}
{Is there a subfamily~$\F' \subseteq \F$ consisting of at most~$k$ sets such that every element of~$U$ is contained in exactly one subset of~$\F'$?}

Observe that since all subsets in~$\F$ have size exactly~$d$, the requirement that each universe element is contained in exactly one subset in~$\F'$ implies that a set~$\F'$ can only be a solution if it consists of~$n/d$ subsets. This implies that~$k = n/d$ for all nontrivial instances of the problem. Hermelin and Wu~\cite{HermelinW12} obtained a compression lower bound for \XSC. The same problem was also studied by Dell and Marx~\cite{DellM12} under the name \textsc{Perfect $d$-Set Matching}. They obtained a slightly stronger compression lower bound, which forms the starting point for our reduction.

\begin{theorem}[{\cite[Theorem 1.2]{DellM12}}] \label{theorem:xsc:nocompression}
For every fixed~$d \geq 3$ and~$\epsilon > 0$, there is no compression of size~$\Oh(k^{d - \epsilon})$ for \XSC unless \containment.
\end{theorem}

We remark that, while Dell and Marx stated their main theorem in terms of kernelizations, the same lower bounds indeed hold for compressions. This follows from the fact that the lower bound machinery on which their result is based holds for arbitrary compressions, rather than just kernelizations (see~\cite{DellM14}). Hermelin and Wu explicitly mention that their (slightly weaker) lower bound also holds against compressions~\cite[\S 1.1]{HermelinW12}.

\subsection{The construction}
We present the construction that will be used to prove that various families of \FDeletion problems do not admit uniformly polynomial kernels. We start by giving some simple propositions that will be used in the correctness proof of the construction. Recall that a vertex~$v$ of a graph~$G$ is \emph{simplicial} if~$N_G(v)$ is a clique. We use~$\tw(G)$ to denote the treewidth of a graph~$G$.

\begin{proposition}[{cf.~\cite[Rule 3.1]{BodlaenderJK13a}}] \label{proposition:tw:onesimplicial}
If~$G$ is a graph and~$v$ is a simplicial vertex of~$G$ of degree at most~$d-1$, then~$\tw(G) \leq d - 1$ if and only if~$\tw(G - \{v\}) \leq d - 1$.
\end{proposition}

Since a $d$-vertex graph has treewidth at most~$d-1$, Proposition~\ref{proposition:tw:onesimplicial} implies the following.

\begin{proposition} \label{proposition:tw:simplicialset}
If~$G$ is a graph and~$Z \subseteq V(G)$ is a set of size~$d$ such that all vertices of~$V(G) \setminus Z$ are simplicial and have degree at most~$d-1$, then~$\tw(G) \leq d-1$.
\end{proposition}

Now we give the construction and prove its correctness.

\begin{lemma} \label{lemma:lowerbound:constructgraph}
For every fixed~$d$ there is a polynomial-time algorithm that, given a set~$U$ of size~$n$, an integer~$k$, and a $d$-uniform set family~$\F \subseteq \binom{U}{d}$, computes a graph~$G'$ with vertex cover number~$\Oh(k^2)$ and an integer~$k' \in \Oh(k^2)$, such that:
\begin{enumerate}
	\item If there is a set~$S' \subseteq V(G')$ of size at most~$k'$ such that~$G' - S'$ is~$K_{d+1}$-minor-free, then there is an exact set cover of~$U$ consisting of~$k$ sets from~$\F$.\label{ppt:modulator:into:cover}
	\item If there is an exact set cover of~$U$ consisting of~$k$ sets from~$\F$, then there is a set~$S' \subseteq V(G')$ of size at most~$k'$ such that~$G' - S'$ is~$K_{d+1}$-minor-free,~$P_{4d}$-minor-free, and has treewidth at most~$d-1$.\label{ppt:cover:into:modulator}
\end{enumerate}
\end{lemma}
\begin{proof}
Given~$U$ of size~$n$, the integer~$k$, and the $d$-uniform set family~$\F$ the algorithm proceeds as follows. If~$k \neq n/d$ then, since precisely~$n/d$ different $d$-size sets are needed to exactly cover~$U$, no exact set cover with~$k$ sets exists. We may then output~$G' := K_{d+1}$ and~$k' := 0$, so we focus on the case that~$k = n/d$. The main idea behind the construction is to create an~$n \times k$ matrix with one vertex per cell. Each one of the~$k$ columns contains~$n$ vertices that correspond to the~$n$ universe elements. By turning columns into cliques and adding small gadgets, we will ensure that solutions to the vertex deletion problem must take the following form: they delete all vertices of the matrix except for exactly~$d$ per column. By enforcing that from each row, all vertices but one are deleted, and that the~$d$ surviving vertices in a column form a subset in~$\F$, we relate the minor-free deletion sets to solutions of the exact covering problem.

The formal construction proceeds as follows. Without loss of generality we can assume that the universe~$U$ consists of~$[n] = \{1, 2, \ldots, n\}$, which simplifies the exposition.

\begin{enumerate}
	\item Initialize~$G'$ as the graph consisting of~$n \times k$ vertices~$v_{i,j}$ for~$i \in [n]$ and~$j \in [k]$. For each column index~$j \in [k]$ turn the vertex set~$\{v_{i,j} \mid i \in [n]\}$ into a clique. We refer to~$M := \{v_{i,j} \mid i \in [n], j \in [k]\}$ as the \emph{matrix vertices}.
	\item For every row index~$i \in [n]$ add a dummy clique~$D_i$ consisting of~$d-1$ vertices to~$G'$. Make all vertices in~$D_i$ adjacent to vertices~$\{v_{i,j} \mid j \in [k]\}$ of the $i$-th row.
	\item As the last step we encode the set family~$\F$ into the graph. For every set~$X \in \binom{U}{d} \setminus \F$, which is a size-$d$ subset of~$[n]$ that is not in the set family~$\F$, we do the following. For each column index~$j \in [k]$, we create an \emph{enforcer vertex~$f_{j,X}$} for the set~$X$ into column~$j$. The neighborhood of~$f_{j,X}$ consists of the~$d$ vertices~$\{v_{i,j} \mid i \in X\}$, i.e., the vertices in column~$j$ corresponding to set~$X$.
\end{enumerate}

\begin{observation} \label{obs:simplicialoutsideM}
All vertices of~$V(G') \setminus M$ are simplicial in~$G'$: their neighborhood is a clique.
\end{observation}

\begin{observation} \label{obs:smallVC}
The set~$M \cup (\bigcup_{i \in n} D_i)$ is a vertex cover of~$G'$ of size~$n(k + d) \in \Oh(k^2)$.
\end{observation}

This concludes the construction of~$G'$. It is easy to see that it can be performed in polynomial time for fixed~$d$, since~$G'$ has~$\Oh(n^{d+1})$ vertices. Define~$k' := k(n-d)$. Since~$d$ is fixed we may absorb it into the $\Oh$-notation. As~$n = kd$ this implies~$k' \in \Oh(k^2)$. We prove that this choice of~$G'$ and~$k'$ satisfies the two statements in the lemma.

\textbf{(\ref{ppt:modulator:into:cover})} Suppose that there is a set~$S' \subseteq V(G')$ of size at most~$k'$ such that~$G' - S'$ is~$K_{d+1}$-minor-free. The following claim shows that~$S'$ intersects~$M$ in a very specific way.

\begin{claim} \label{claim:solution:selects:matrix}
The set~$S'$ is a subset of~$M$ that contains exactly~$k-1$ vertices from each row and~$n-d$ vertices from each column.
\end{claim}
\begin{claimproof}
Since each of the~$k$ columns of~$M$ induces a clique in~$G$,~$S'$ avoids at most~$d$ vertices in each column. Since each column contains~$n$ vertices, this implies that~$S'$ contains at least~$n - d$ vertices from each column, so at least~$k(n-d) = k'$ vertices from~$M$. Hence the set~$S'$ of size~$k'$ cannot contain any other vertex and must be a subset of~$M$. If there is a column from which~$S'$ contains more than~$n - d$ vertices, then together with the~$n - d$ vertices from each of the remaining~$k - 1$ columns the size of~$S'$ is at least~$n - d + 1 + (k-1)(n - d) = k' + 1$, contradicting the choice of~$S'$. Hence~$S'$ contains exactly~$n-d$ vertices from each column of~$M$.

The argumentation for rows is similar, but here we also use the dummy cliques~$D_i$ for~$i \in [n]$. If~$S'$ contains less than~$k-1$ vertices from the $i$-th row of~$M$ (i.e., of~$\{v_{i,j} \mid j \in [k]$), then two remaining vertices in row~$i$ together with~$D_i$ of size~$d-1$ form a clique of size~$d+1$, contradicting our choice of~$S'$. Suppose there is a row from which~$S'$ contains more than~$k-1$ vertices. Since~$S'$ also contains at least~$k-1$ vertices from each of the other~$n - 1$ rows, this implies that the total size of~$S'$ is at least~$k + (n-1)(k-1) = k + nk - n - k + 1 = nk - dk + 1 = k' + 1$, where the last step uses the fact that~$dk = n$ as observed at the beginning of the proof. Hence the set~$S'$ of size~$k'$ contains exactly one vertex from each row.
\end{claimproof}

\begin{claim}
Let~$j \in [k]$ and let~$X_j := \{i \mid i \in [n] \wedge v_{i,j} \not \in S'\}$. Then~$X_j \in \F$.
\end{claim}
\begin{claimproof}
By the previous claim, the set~$X_j$ has size exactly~$d$. To see that the set family~$\F$ indeed has a set containing the universe elements corresponding to the vertices in column~$j$ that are avoided by the deletion set~$S'$, observe the following. If~$X_j \not \in \F$ then during the construction we created an enforcer vertex~$f_{j,X_j}$ for set~$X_j$ into column~$j$. But then vertex~$f_{j,X_j}$, which is not contained in~$M$ and therefore not in~$S'$, forms a clique together with its~$d$ neighbors in column~$j$. As the size of this clique is~$d+1$, this contradicts the choice of~$S'$.
\end{claimproof}

Using the two claims it is easy to finish the proof. For each~$j \in [k]$, define the set~$X_j$ as in the second claim. It follows that the subsets~$X_1, \ldots, X_k$ are contained in~$\F$. Since~$S'$ avoids exactly one element in each row by Claim~\ref{claim:solution:selects:matrix}, no universe element is contained in two different sets~$X_j, X_{j'}$. To see that every universe element is contained in at least one set~$X_j$, note that the previous argument shows that the~$k$ size-$d$ sets~$X_1, \ldots, X_k$ are pairwise disjoint. Their union therefore has size~$dk = n$, which proves that all universe elements are covered. 

\textbf{(\ref{ppt:cover:into:modulator})} It remains to prove the second statement in the lemma. Suppose that~$\F' \subseteq \F$ is an exact set cover of~$U$. As observed above,~$\F'$ consists of~$k$ distinct subsets~$X_1, \ldots, X_k \subseteq [n]$, each of size~$d$. We construct a deletion set~$S' \subseteq V(G')$ as follows. For each column index~$j \in [k]$, add the vertices~$\{v_{i,j} \mid i \not \in X_j\}$ to~$S'$. Clearly the resulting set~$S'$ has size exactly~$k' = (n-d)k$. Since a graph of treewidth at most~$d - 1$ does not contain~$K_{d+1}$ as a minor~\cite{Bodlaender98}, it suffices to prove that~$G' - S'$ has treewidth at most~$d-1$ and avoids~$P_{4d}$ as a minor. Before proving these two claims, we consider the structure of the connected components of~$G' - S'$. For each column~$j \in [k]$ define~$Z_j := \{v_{i,j} \mid i \in X_j\}$, which are precisely the vertices in column~$j$ not contained in~$S'$. By the construction of~$G'$ they induce a clique in~$G'$, and are therefore contained in a single connected component of~$G' - S'$.

\begin{claim} \label{claim:ppt:componentcapm}
Let~$C_j$ be the connected component of~$G' - S'$ containing~$Z_j$. Then~$C_j \cap M = Z_j$.
\end{claim}
\begin{claimproof}
Assume for a contradiction that~$C_j$ contains a vertex~$v \in M \setminus Z_j$. Let~$P$ be a shortest path from a member of~$Z_j$ to~$v$ through~$C_j$. Then~$P$ is an induced path. Suppose that~$P$ contains a vertex of~$V(G') \setminus M$. Then~$P$ has at least three vertices and contains a vertex~$u$ of~$V(G') \setminus M$ as an interior vertex. But by Observation~\ref{obs:simplicialoutsideM} vertex~$u$ is simplicial in~$G'$ and therefore in~$G' - S'$; this contradicts the fact that~$P$ is an induced path. We may conclude that~$P$ consists entirely of vertices of~$M$. 

Observe that by construction of~$G'$, the vertices in~$N_{G'}(Z_j) \cap M$ are those in~$M \setminus Z_j$ that share a row or column with a member of~$Z_j$. By our choice of~$Z_j$, all vertices of~$M \setminus Z_j$ that are in column~$j$ together with~$Z_j$, are contained in~$S'$ and therefore do not occur in the component~$C_j$ of~$G' - S'$. Now consider vertices that share a row with~$Z_j$. Since we constructed~$S'$ from an exact set cover of~$U$, every element of~$U$ is contained in exactly one of the sets~$X_1, \ldots, X_k$. This implies that for each~$i \in [n]$ such that vertex~$v_{i,j} \in Z_j$, we have~$v_{i, j'} \in S'$ for~$j' \neq j$. Hence all vertices not in~$Z_j$ that share a row with a member of~$Z_j$, are contained in~$S'$ and do not occur in the connected component~$C_j$ of~$G' - S'$. It follows that~$N_{G' - S'}(Z_i) \cap M = \emptyset$ and therefore that~$P$ cannot be an induced path in~$(G' - S')[M]$ connecting a vertex of~$Z_i$ and a vertex of~$M \setminus Z_i$. Hence~$C_j \cap M = Z_j$.
\end{claimproof}

\begin{claim} \label{claim:ppt:simplicialsmalldegree}
Let~$C_j$ be the connected component of~$G' - S'$ containing~$Z_j$. Then all vertices of~$C_j \setminus Z_j$ are simplicial in~$G$ (and therefore in~$C_j$) and have degree less than~$d$ in~$C_j$.
\end{claim}
\begin{claimproof}
By the previous claim, the component~$C_j$ consists of~$Z_j$ together with vertices in~$V(G') \setminus M$. By Observation~\ref{obs:simplicialoutsideM}, all vertices of~$C_j \setminus M$ are simplicial in~$G'$ and therefore in~$C_j$. Consider a simplicial vertex~$v \in C_j \setminus M$. If~$v$ is a vertex in the dummy clique for a row~$i$ ($i \in [n]$), then the neighborhood of~$v$ in~$C_j$ consists of~$D_i \setminus \{v\}$ together with the single vertex in row~$i$ that is not in~$S'$. Since~$D_i$ has size~$d-1$, the degree of~$v$ is~$((d-1) - 1) + 1 < d$. If~$v$ is not in a dummy clique, then~$v$ is an enforcer vertex for some set~$X \in \binom{U}{d} \setminus \F$ into column~$j$. Recall that~$Z_j$ contains the vertices of the $j$-th column corresponding to the set~$X_j \in \F$ in the exact cover. Since~$X_j \in \F$ while~$X \not \in \F$, it follows that at least one neighbor of~$v$ in~$G$ is not contained in~$Z_j$. Hence the degree of~$v$ in~$C_j$ is strictly less than~$d$.
\end{claimproof}

The following claim summarizes our insights into the structure of~$G' - S'$.

\begin{claim} \label{claim:ppt:resultingcomponents}
Every connected component of~$G' - S'$ is either a singleton enforcer vertex, or a component~$C_j$ for~$j \in [k]$ consisting of the vertices~$Z_j$, the dummy vertices~$\bigcup_{i:v_{i,j} \in Z_j} D_i$, and the enforcer vertices into column~$j$ for sets~$X$ that intersect~$X_j$.
\end{claim}
\begin{claimproof}
Consider a connected component~$C$ of~$G' - S'$ that is not a single enforcer vertex. Then~$C$ contains at least one vertex from~$M$. If~$C$ contains an enforcer vertex then this follows from the fact that the only neighbors of enforcer vertices are in~$M$. If~$C$ contains a dummy vertex of clique~$D_i$, then let~$j$ be the index of the set covering~$i$ (such that~$i \in X_j)$; it follows that~$v_{i,j} \not \in S'$ is a neighbor of the dummy vertex that is contained in~$M$ and belongs to the same connected component. Hence every~$C$ contains at least one vertex from~$M$; let~$j$ be a column containing such a vertex. Then all of~$Z_j$ is contained in~$C$, since~$Z_j$ forms a clique. Each vertex from~$Z_j$ is adjacent to the dummy clique in its row. Since no vertices from dummy cliques are contained in~$S'$, this shows that~$\bigcup _{i: v_{i,j} \in Z_j} D_i$ is also contained in the component~$C$ of~$G' - S'$. If~$f_{j,X}$ is an enforcer vertex for a set~$X$ with~$X \cap X_j \neq \emptyset$, then~$f_{j,X}$ is a neighbor of any vertex representing a vertex in~$X \cap X_j$, showing that~$f_{j,X}$ is in component~$C$. Hence~$C$ includes all the vertices mentioned in the claim. To see that~$C$ cannot include any other vertex, observe that by Claim~\ref{claim:ppt:componentcapm} component~$C$ contains no other vertex of~$M$. Since the dummy vertices for the remaining rows are only adjacent to vertices of~$M \setminus Z_j$, they are not contained in component~$C$. The same holds for enforcer vertices into columns that are not~$j$. Finally, it is easy to verify that enforcer vertices into column~$j$ for sets~$X$ that are disjoint from~$X_j$ form singleton connected components of~$G' - S'$ and are not included in~$C$. This proves the claim.
\end{claimproof}

Let us prove that~$G' - S'$ has treewidth at most~$d - 1$. Since singleton graphs have treewidth zero, by the previous claim it suffices to bound the treewidth of components~$C_j$ of the form described in the claim. By Claim~\ref{claim:ppt:simplicialsmalldegree} all vertices of~$C_j \setminus Z_j$ are simplicial in~$C_j$ and have degree less than~$d$ in~$C_j$. Since the set~$Z_j$ has size~$d$, by Proposition~\ref{proposition:tw:simplicialset} we now find that~$\tw(C_j) \leq d - 1$. As this bounds the treewidth of all nontrivial components of~$G' - S'$, this proves that~$\tw(G' - S') \leq d-1$ and therefore that~$G' - S'$ is~$K_{d+1}$-minor-free. It remains to prove that~$G'$ has no path minor on~$4d$ vertices or more.

\begin{claim} \label{claim:ppt:nopathsubgraph}
No connected component of~$G' - S'$ contains a simple path of~$4d$ vertices.
\end{claim}
\begin{claimproof}
As the claim trivially holds for singleton components, it suffices to consider components~$C_j$ for~$j \in [k]$ that contain one of the sets~$Z_j$, as described in Claim~\ref{claim:ppt:resultingcomponents}. Consider a simple path~$P$ in~$C_j$. Since the only neighbors of the enforcer vertices in~$C_j$ are the~$d$ vertices in~$Z_j$, path~$P$ must visit a vertex of~$Z_j$ between visiting different enforcer vertices. Since no vertex of~$Z_j$ can be visited twice,~$P$ contains at most~$d+1$ enforcer vertices. To see that~$P$ contains at most~$2(d-1)$ dummy vertices, we will prove that~$P$ contains vertices from at most two different dummy cliques~$D_i, D_{i'}$. To see this, observe that the vertex~$v_{i,j}$ is a cut-vertex in~$C_j$ that separates the dummy clique~$D_i$ from the rest of the component~$C_j$. Hence the path cannot enter a dummy clique, visit some vertices, and then exit to the rest of the component. Hence at most two dummy cliques can be visited by~$P$, one containing the starting point of~$P$ and one containing the endpoint of~$P$. So indeed~$P$ contains at most~$2(d-1)$ dummy vertices. By Claim~\ref{claim:ppt:resultingcomponents} the only vertices in~$C_j$ are those of dummy cliques, the enforcer vertices, and the set~$Z_j$. Since~$P$ contains at most~$2(d-1)$ dummies,~$d+1$ enforcers, and the set~$Z_j$ has size~$d$, it follows that the simple path~$P$ contains at most~$2(d-1) + (d+1) + d = 4d - 1$ vertices.
\end{claimproof}

Since~$G$ has a path on~$d$ vertices as a minor if and only if it contains a path on~$d$ vertices as a subgraph, Claim~\ref{claim:ppt:nopathsubgraph} proves that~$G' - S'$ is~$P_{4d}$-minor-free. This concludes the proof of Lemma~\ref{lemma:lowerbound:constructgraph}.
\end{proof}

\subsection{Kernelization lower bounds}
By combining the construction of Lemma~\ref{lemma:lowerbound:constructgraph} with the tools of Section~\ref{section:lowerbound:tools} we now derive several kernelization lower bounds for \FDeletion problems. Concretely, we prove Theorem~\ref{theorem:cliqueminordeletion:lowerbound}.

\begin{untheorem}
Let~$d \geq 3$ be a fixed integer and~$\epsilon > 0$. If the parameterization by solution size~$k$ of one of the problems
\begin{enumerate}
	\item \KdplusoneDeletion, 
	\item \CliqueorPathDeletion, and
	\item \TreewidthDminoneDeletion 
\end{enumerate}
admits a compression of bitsize~$\Oh(k^{\frac{d}{2} - \epsilon})$, or a kernel with $\Oh(k^{\frac{d}{4} - \epsilon})$ vertices, then \containment. In fact, even if the parameterization by the size~$x$ of a vertex cover of the input graph admits a compression of bitsize~$\Oh(x^{\frac{d}{2} - \epsilon})$ or a kernel with $\Oh(x^{\frac{d}{4} - \epsilon})$ vertices, then \containment.
\end{untheorem}
\begin{proof}
We first prove the statement about \KdplusoneDeletion. Observe that for every fixed~$d$, the transformation of Lemma~\ref{lemma:lowerbound:constructgraph} forms a degree-two polynomial-parameter transformation from \XSC to \KdplusoneDeletion: it is a polynomial-time algorithm that maps an instance~$(U,\F,k)$ of \XSC to an instance~$(G',k')$ of \KdplusoneDeletion with~$k' \in \Oh(k^2)$, and the two statements in the lemma ensure that~$(U,\F,k)$ is a \yes-instance if and only if~$(G',k')$ is a \yes-instance.

Now assume that \KdplusoneDeletion parameterized by~$k$ has a compression of size~$\Oh(k^{\frac{d}{2} - \epsilon})$ for some~$d \geq 3$ and~$\epsilon > 0$. By Proposition~\ref{proposition:compressionbyppt}, it follows that \XSC has a compression of size~$\Oh(k^{d - 2\epsilon})$. By Theorem~\ref{theorem:xsc:nocompression} this implies \containment. Observe that, since the graph constructed in Lemma~\ref{lemma:lowerbound:constructgraph} has a vertex cover of size~$\Oh(k^2)$, the construction also serves as a degree-two polynomial parameter transformation from \XSC to the parameterization of \KdplusoneDeletion by vertex cover. Hence the existence of a compression with bitsize~$\Oh(x^{\frac{d}{2} - \epsilon})$ implies \containment by the same argument as above.

Concerning the existence of kernels with few vertices, observe that a kernelized instance with~$\Oh(k^{\frac{d}{4} - \epsilon})$ vertices can be encoded in~$\Oh(k^{\frac{d}{2} - 2\epsilon})$ bits, by writing down the adjacency matrix of the graph and target value (which does not exceed the order of the graph) in binary. Hence a kernel with~$\Oh(k^{\frac{d}{4} - \epsilon})$ or $\Oh(x^{\frac{d}{4} - \epsilon})$ vertices yields a compression that implies \containment.

Now consider the other two problems mentioned in the theorem. By exactly the same argumentation, it suffices to argue that the construction of Lemma~\ref{lemma:lowerbound:constructgraph} is a valid degree-two polynomial parameter transformation of \XSC into these problems. Let~$(U,\F,k)$ be an instance of \XSC and consider the pair~$(G',k')$ constructed in the lemma. If~$(U,\F,k)$ is a \yes-instance, then by the second statement of Lemma~\ref{lemma:lowerbound:constructgraph} we can delete~$k'$ vertices from~$G'$ to make it both~$K_{d+1}$ and~$P_{4d}$ minor-free, hence~$(G',k')$ is a \yes-instance of \CliqueorPathDeletion. For the reverse direction, if~$G'$ can be made both~$K_{d+1}$ and~$P_{4d}$ minor-free by~$k'$ vertex deletions, then in particular it can be made~$K_{d+1}$-minor-free by~$k'$ deletions so the first item of Lemma~\ref{lemma:lowerbound:constructgraph} proves that~$(U,\F,k)$ is a \yes-instance. So the construction is a degree-two polynomial parameter transformation to \CliqueorPathDeletion.

Finally, consider the \TreewidthDminoneDeletion problem. If~$(U,\F,k)$ is a \yes-instance then, by the second item of Lemma~\ref{lemma:lowerbound:constructgraph}, the treewidth of the constructed graph~$G'$ can be reduced to at most~$d-1$ by~$k'$ deletions, implying that~$(G',k')$ is a \yes-instance of \TreewidthDminoneDeletion. For the reverse direction, if the treewidth of~$G'$ can be reduced to at most~$d-1$ by~$k'$ deletions, then since a graph of treewidth at most~$d-1$ does not contain~$K_{d+1}$ as a minor (cf.~\cite{Bodlaender98}), it follows that~$G'$ can be made~$K_{d+1}$-minor-free by~$k'$ deletions. By the first item of Lemma~\ref{lemma:lowerbound:constructgraph}, this implies that~$(U, \F, k)$ is a \yes-instance. The construction is therefore also a degree-two polynomial parameter transformation to \TreewidthDminoneDeletion, which proves the theorem.
\end{proof}
\section{Structural results about treedepth}
In this section we derive several properties of treedepth decompositions that will be needed to analyze the effect of the graph reduction steps. We start by proving some general facts about treedepth in Section~\ref{section:treedepth:properties}. In Section~\ref{section:nearly:clique:isolated} we start analyzing properties of instances of \TreedepthEtaDeletion, introducing the notion of nearly clique separated sets to show that minimum solutions intersect certain parts of the graph in only few vertices. Finally, in Section~\ref{section:treedepth:transformations} we present the lemmata discussed in the introduction concerning three types of graph transformations (edge additions, edge removals, and vertex removals) and derive conditions under which these do not change the answer to an instance of \TreedepthEtaDeletion.

\subsection{Properties of treedepth} \label{section:treedepth:properties}

The following lemma shows that either the reach of a vertex is large, or the height of the decomposition is large. It will be used to argue that, in treedepth-$\eta$ decompositions of a reduced graph, the reach of a vertex is large enough to allow a deleted component of the graph to be embedded below it without increasing the total decomposition height.

\begin{lemma}\label{lem:neighheght}
Let~$d$ and~$t$ be positive integers,~$G$ be a connected graph, and let~$H_{1},H_{2},\dots, H_{t}$ be vertex-disjoint connected subgraphs of~$G$ with~$\td(H_{i})\geq d$. Let~$T$ be a treedepth decomposition of~$G$. If~$v \in V(G)$ such that~$v\in N_{G}(H_{i})$ for every $i\in [t]$, and~$\reach(v,T) < d$, then $\height(T)\geq t+1$.
\end{lemma}
\begin{proof}
Suppose that~$\reach(v,T) = \height(T) - \depth(v,T) < d$. It follows that the subtree~$T_v$ rooted at~$v$ has height at most~$d$, otherwise the path from the root of~$T$ to~$v$, and then to a deepest leaf in~$T_v$, would contain at least~$\depth(v,T) + (d + 1) - 1$ vertices (we subtract one because~$v$ is counted twice). Since~$\height(T) = \reach(v,T) + \depth(v,T) < d + \depth(v,T)$, this would give a contradiction.

\begin{claim}
For every~$i \in [t]$ some vertex of~$H_i$ is not in~$T_v$.
\end{claim}
\begin{claimproof}
If~$H_i \subseteq T_v$ then the rooted subtree~$T_v$ is a treedepth decomposition of~$G[T_v]$, a supergraph of~$H_i$, of height~$d$. Since vertex~$v$ is not in~$H_i$, the rooted forest obtained by removing~$v$ is a treedepth decomposition of~$G[T_v \setminus \{v\}]$ of height less than~$d$. But then the supergraph~$G[T_v \setminus \{v\}]$ of~$H_i$ has treedepth less than~$d$, contradicting~$\td(H_i) \geq d$.
\end{claimproof}

\begin{claim}
For every~$i \in [t]$ some vertex of~$H_i$ is an ancestor of~$v$.
\end{claim}
\begin{claimproof}
Consider some~$i \in [t]$. Since~$v \in N_G(H_i)$, there is a vertex~$u$ in~$H_i$ that is adjacent to~$v$ and therefore~$u$ is an ancestor of~$v$ in~$T$, or~$u$ is in~$T_v$. In the first case we are done. In the second case, let~$w$ be a vertex of~$H_i$ that is not in~$T_v$, which exists by the previous claim. If~$w$ is an ancestor of~$v$ we are again done. If not, then~$u$ and~$v$ are vertices that are not in ancestor-descendant relation that belong to the same connected subgraph~$H_i$ of~$G$. Hence, by Observation~\ref{observation:connectedsubgraph:ancestorpath} there is a vertex in~$H_i$ that is a common ancestor of~$u$ and~$v$. Since all ancestors of~$u \in T_v$ are ancestors of~$v$, the claim follows.
\end{claimproof}

The claim shows that~$v$ has~$t$ ancestors unequal to~$v$ itself. Hence~$\height(T) \geq t + 1$.
\end{proof}

The next lemma will be used to argue that an edge must be represented in sufficiently shallow decompositions.

\begin{lemma}\label{lem:bgflw}
Let $q$ be a positive integer and $G$ be a graph. If $u,v\in V(G)$ are joined by $q$ internally vertex-disjoint paths and $F$ is a treedepth decomposition of $G$ in which~$u$ and~$v$ are not in ancestor-descendant relation, then $\height(F)>q$.
\end{lemma}

\begin{proof}
Let~$P_u$ and~$P_v$ be the paths from~$u$ and~$v$ to the root, and let~$P$ be the intersection of the two paths. Since~$u$ and~$v$ are not in ancestor-descendant relation we have~$u,v \not \in P$. As $P$ contains all common ancestors of $u$ and $v$ in $F$, in graph~$G$ the vertices~$u$ and~$v$ are separated by $V(P)$. From Menger's Theorem, as $u$ and $v$ are connected by $q$ internally vertex-disjoint paths, it follows that $|V(P)|\geq q$. This implies that $\height(T)>q$.
\end{proof}

The following technical lemma gives conditions under which a treedepth decomposition can be modified to ensure that a vertex set~$V(H)$ is embedded in the subtree below a distinguished vertex~$v$. The updated decomposition therefore represents all possible edges between~$v$ and~$V(H)$. The lemma will be crucial in the correctness proof of Lemma~\ref{lemma:remove:edges}, which gives conditions under which edges can safely be deleted from an instance.

\begin{figure}[t]
\begin{center}
\subfigure[Graph~$G$, $H \subseteq G$.]{\label{fig:switch:graph}
\image{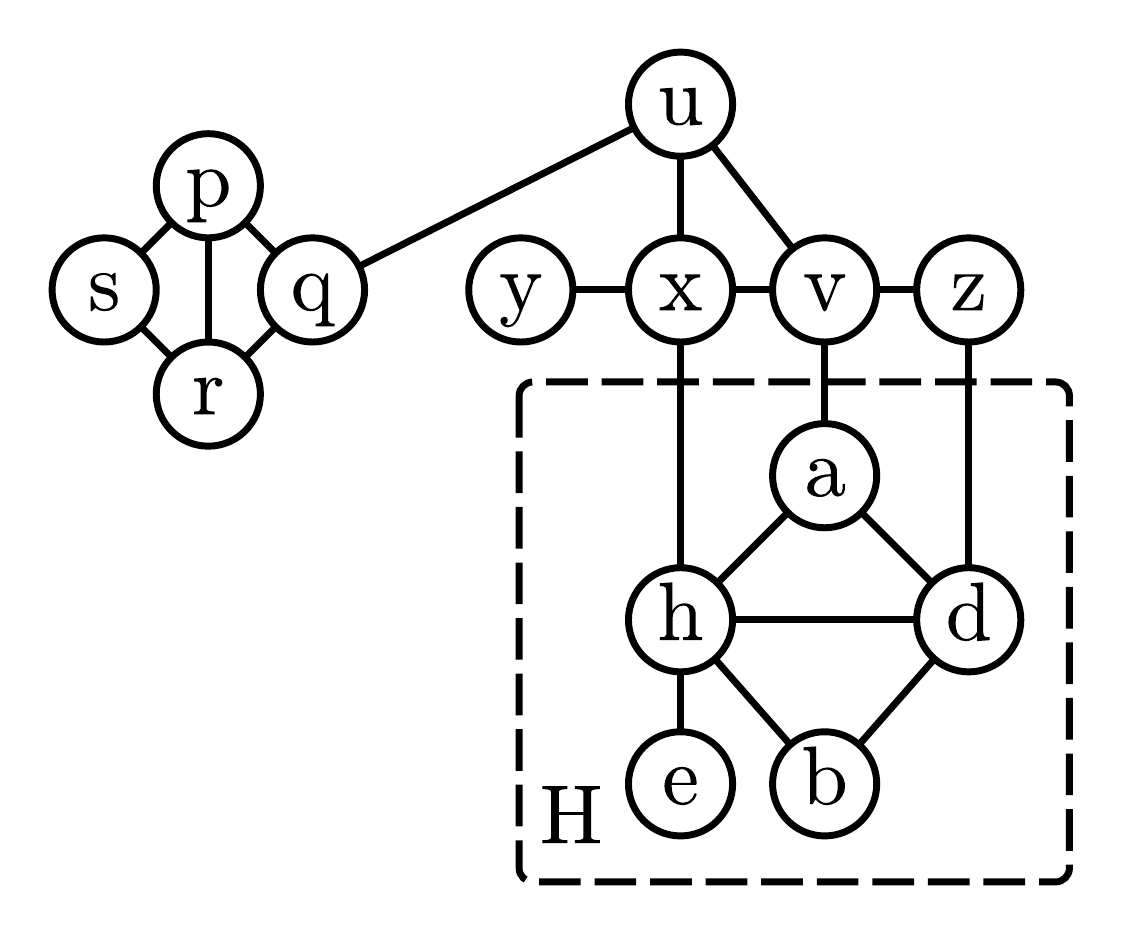}
}
\subfigure[$T = T^*$.]{\label{fig:switch:nicedec}
\image{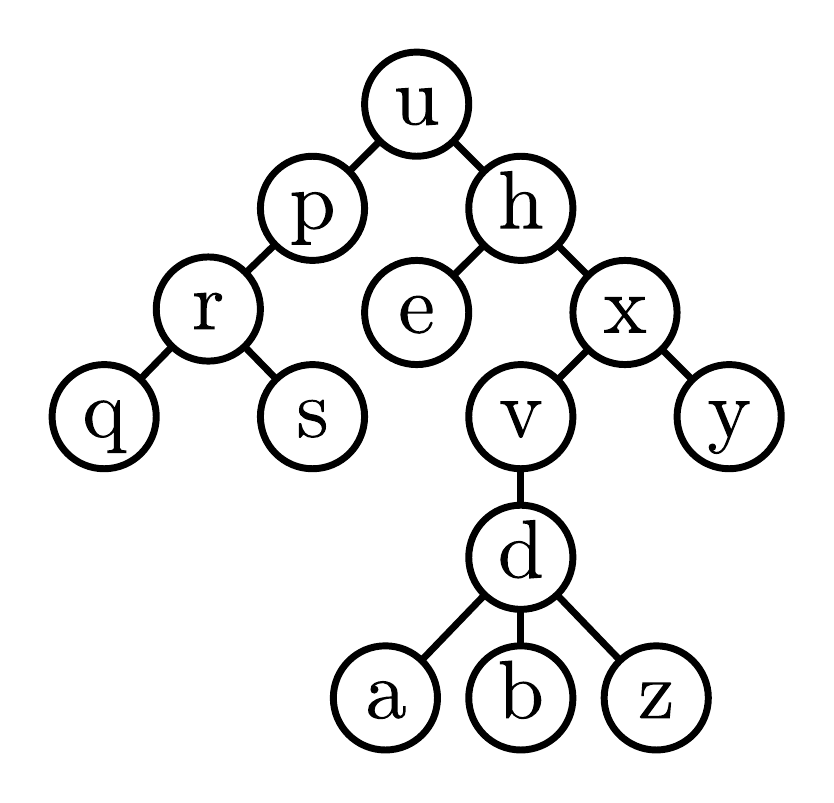}
}
\subfigure[$\hat{T}$.]{\label{fig:switch:loweredtrees}
\image{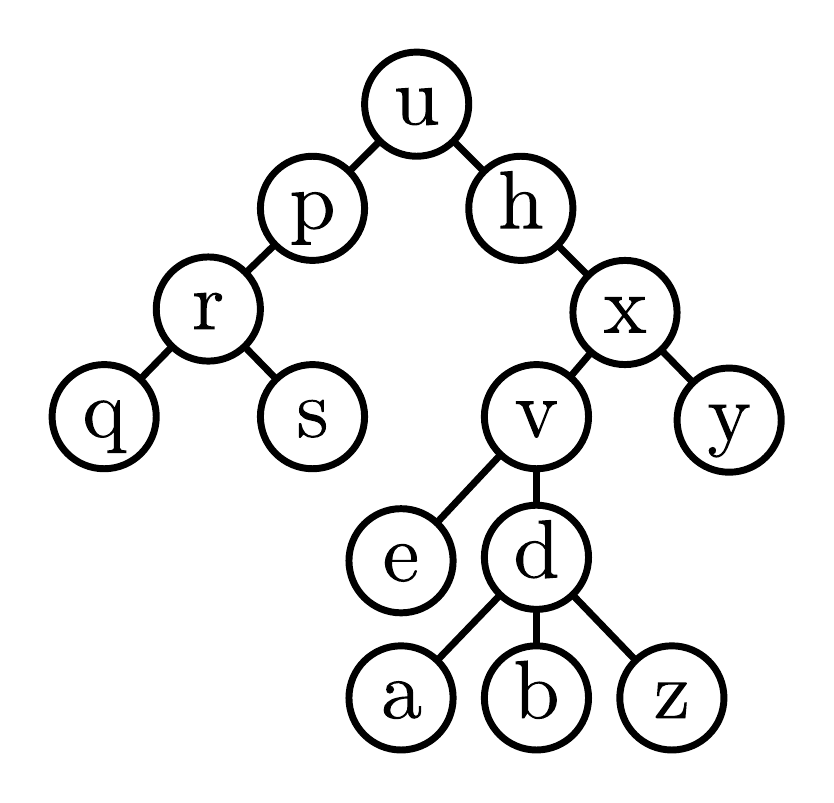}
}
\subfigure[$\widetilde{T}$.]{\label{fig:switch:switched}
\image{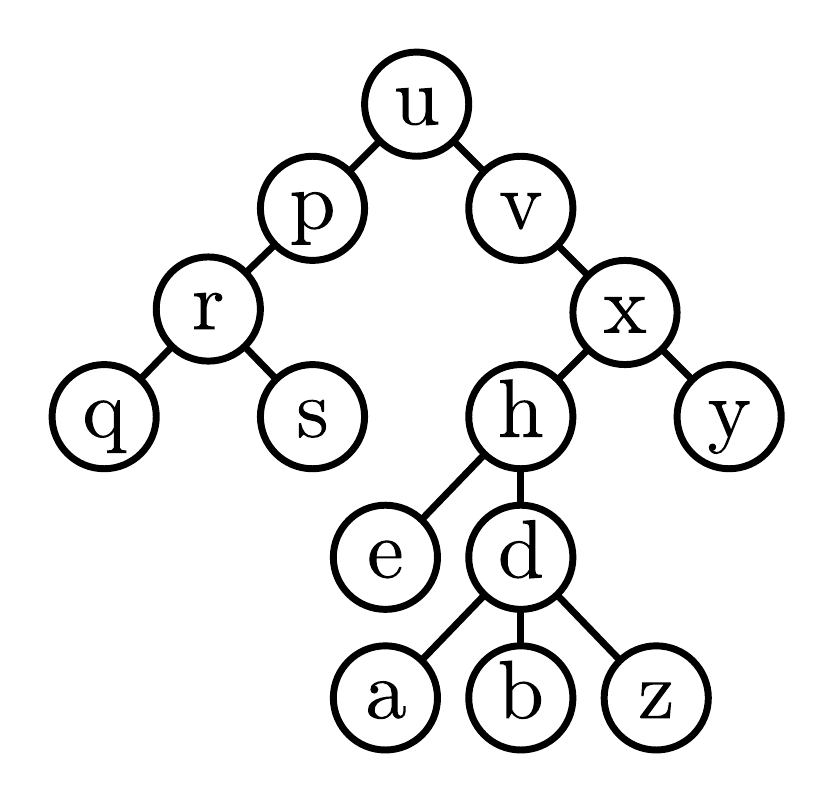}
}
\caption{Illustration for the proof of Lemma~\ref{lemma:switch}. \ref{fig:switch:graph} The graph~$G$ and connected subgraph~$H$ consisting of~$\{a,b,h,d,e\}$ are shown. As~$N_G(H) = \{x,v,z\} \subseteq N_G[v] = \{u,v,x,z,a\}$, the lemma applies. \ref{fig:switch:nicedec} A treedepth decomposition~$T$ for~$G$. As it is a nice decomposition, it is also used as~$T^*$. Vertex~$h$ is the first member of~$H$ on the path from~$v$ to the root. Vertex~$z \in N_G(H)$ is a neighbor of~$v$ that is not an ancestor of~$v$. The path~$P$ consists of~$(x,h)$. \ref{fig:switch:loweredtrees} The result of cutting off the (singleton) subtree~$T^*_e$ and attaching an minimum-height decomposition~$\hat{T}^e$ for that subgraph at~$v$. \ref{fig:switch:switched} The final decomposition~$T'$ used to invoke the induction hypothesis.}
\end{center}
\end{figure}

\begin{lemma} \label{lemma:switch}
Let~$G$ be a connected graph, let~$H \subseteq G$ be a connected subgraph of~$G$, and let~$v \in V(G) \setminus V(H)$ be a vertex such that~$N_G(H) \subseteq N_G[v]$. For any treedepth decomposition~$T$ of~$G$, there exists a treedepth decomposition~$T'$ of~$G$ such that:
\begin{enumerate}
	\item $\height(T') \leq \max (\height(T), \depth(v, T) + \td(G[V(H)]))$.
	\item All vertices of~$V(H)$ belong to~$T'_v$, the subtree of~$T'$ rooted at~$v$.
\end{enumerate}
\end{lemma}
\begin{proof}
Let~$G, H$, and~$v$ be as stated. We use induction on~$\depth(v,T)$. If~$\depth(v,T) = 1$ then~$v$ is the root of~$T$. Since~$G$ is connected all its vertices belong to the same decomposition tree and are therefore contained in~$T_v$. So~$T' = T$ trivially satisfies the requirements. For the induction step, assume that~$\depth(v,T) > 1$.

By Proposition~\ref{proposition:nicedec:nothigher} there exists a nice treedepth decomposition~$T^*$ of~$G$ whose height does not exceed the height of~$T$, such that no vertex has greater depth in~$T^*$ than in~$T$. If~$V(H) \subseteq T^*_v$ then the lemma holds, as we may take~$T'$ equal to~$T^*$. Assume for the remainder that~$V(H) \setminus T^*_v \neq \emptyset$. Since~$G$ is connected and~$V(H) \subsetneq V(G)$ (as~$v \in V(G) \setminus V(H)$), the set~$N_G(H)$ is not empty.

First consider the case that all vertices of~$N_G(H)$ lie on the path from~$v$ to the root in~$T^*$. Then it is easy to find a decomposition as described in the lemma: we form~$T'$ by restricting the decomposition~$T^*$ to the vertices of~$V(G) \setminus V(H)$, then we take a minimum-height treedepth decomposition~$T^H$ of the graph~$G[V(H)]$ and attach the root of~$T^H$ to vertex~$v$ to ensure that~$H \subseteq T'_v$. From this construction it easily follows that~$\height(T') \leq \max (\height(T), \depth(v,T) + \td(G[V(H)]))$. To see that all edges are represented in the model, observe that (i) all edges of~$G - V(H)$ are represented because~$T'$ contains the restriction of~$T^*$ to~$V(G) \setminus V(H)$, (ii) all edges of~$G[V(H)]$ are represented because a valid decomposition of~$G[V(H)]$ is inserted into~$T'$, while finally all edges between~$V(H)$ and~$V(G) \setminus V(H)$ are represented because all vertices of~$N_G(H)$ are ancestors of~$v$ and therefore ancestors of every vertex in~$H \subseteq T'_v$. 

In the remainder we therefore focus on the case that some vertex~$z \in N_G(H)$ is not an ancestor of~$v$. Since~$v$ is an ancestor of itself, we have~$z \neq v$. Observe that~$z$ is adjacent in~$G$ to~$v$, since~$z \in N_G(H) \subseteq N_G[v]$ and~$z \neq v$. Since~$T^*$ is a valid treedepth decomposition, if~$z$ is not an ancestor of~$v$, then~$z \in T^*_v$.

\begin{claim}
The path in~$T^*$ from~$v$ to the root contains a vertex of~$H$.
\end{claim}
\begin{claimproof}
Assume for a contradiction that the path from~$v$ to the root contains no vertex of~$H$, i.e., no ancestor of~$v$ is contained in~$H$. Let~$h \in V(H) \setminus T^*_v$, which exists by our assumption above. Let~$h' \in V(H)$ be adjacent in~$G$ to~$z$; such a vertex exists since~$z \in N_G(H)$. If one of~$h$ or~$h'$ is an ancestor of~$v$ in~$T^*$ then we are done. If this is not the case, then observe that~$h' \in T^*_v$: to realize its edge to~$z \in T^*_v$ without being an ancestor of~$v$, it must lie in~$T^*_v$. Since~$h \not \in T^*_v$ but~$h' \in T^*_v$, the only common ancestors of~$h$ and~$h'$ are ancestors of~$v$, which are not contained in~$H$ by assumption. But by Observation~\ref{observation:connectedsubgraph:ancestorpath}, the common ancestors of~$h$ and~$h'$ separate~$h$ and~$h'$ in~$G$. But then these common ancestors are a vertex subset of~$V(G) \setminus V(H)$ that separate~$h$ and~$h'$ in~$G$, contradicting the fact that~$H$ is a connected subgraph of~$G$. The claim follows.
\end{claimproof}

Let~$h \in V(H)$ be the first vertex from~$H$ on the path from~$v$ to the root of~$T^*$. Let~$P$ be the path from~$\pi(v)$ to~$h$ in~$T^*$. This choice of~$h$ and the fact that~$T^*$ is nice has a useful consequence. Let~$c_1, \ldots, c_t$ be vertices of~$T^*$ that are unequal to~$v$, whose parent belongs to~$P$, and for which the subtree rooted there contains at least one vertex of~$V(H)$. (It may be that there are no such vertices.)

\begin{claim} \label{claim:subtrees:no:nh}
For each vertex~$c_i$ with~$i \in [t]$, the subtree~$T^*_{c_i}$ contains no vertex of~$N_G(H)$.
\end{claim}
\begin{claimproof}
Suppose that~$T^*_{c_i}$ contains a vertex~$x \in N_G(H)$. By our definition of~$c_i$ we have~$x \neq v$ and vertex~$c_i$ is not in an ancestor-descendant relation with~$v$ in~$T^*$. Hence no descendant of~$c_i$ is in ancestor-descendant relation with~$v$ either. Since~$N_G(H) \subseteq N_G[v]$ and~$x \neq v$ we have~$\{x,v\} \in E(G)$. However, since~$x \in T^*_{c_i}$ is not in ancestor-descendant relation with~$v$, this edge is not realized in the decomposition~$T^*$; a contradiction.
\end{claimproof}

\begin{claim} \label{claim:move:noinbetween}
For each vertex~$c_i$ with~$i \in [t]$ we have~$T^*_{c_i} \subseteq V(H)$.
\end{claim}
\begin{claimproof}
Assume for a contradiction that~$T^*_{c_i} \not \subseteq V(H)$. Since~$T^*_{c_i}$ contains at least one vertex of~$V(H)$ by definition of~$c_i$, it follows that~$T^*_{c_i}$ contains both a vertex~$h' \in V(H)$ and a vertex~$z \not \in V(H)$. By the definition of a nice treedepth decomposition, the graph~$G[T^*_{c_i}]$ is connected and contains a path~$P'$ from~$h'$ to~$z$. Since~$h'$ is in~$V(H)$ but~$z$ is not, it follows that~$P' \subseteq T^*_{c_i}$ contains a vertex of~$N_G(H)$; a contradiction to Claim~\ref{claim:subtrees:no:nh}.
\end{claimproof}

For each~$i \in [t]$ let~$\hat{T}^i$ be a minimum-height treedepth decomposition of the graph~$G[T^*_{c_i}]$. Since~$T^*_{c_i} \subseteq V(H)$ for all~$i \in [t]$ it follows that~$\height(\hat{T}^i) \leq \td(G[V(H)])$ for all~$i \in [t]$. We now obtain a new decomposition tree~$\hat{T}$ from~$T^*$ as follows. Remove the subtrees~$T^*_{c_1}, \ldots, T^*_{c_t}$ from~$T^*$. Then add the trees~$\hat{T}^1, \ldots, \hat{T}^t$ and connect the root of each of these trees to~$v$. This results in a valid decomposition of~$G$ of height at most~$\max(\height(T^*), \depth(v, T^*) + \td(G[V(H)]))$. To see that the decomposition is valid, observe that all edges of~$G[T^*_{c_i}]$ for~$i \in [t]$ are represented in the subtrees that we inserted. Since we attach the replacement trees to the ancestor~$v$ of the vertex of the path~$P$ that they were originally attached to, the edges to the rest of the graph are represented as well. 

As the next step, we swap the labels of~$h$ and~$v$ in~$\hat{T}$ and use the resulting decomposition as~$\widetilde{T}$. It is easy to see that moving~$v$ to the location of its ancestor~$h$ maintains the fact that all edges incident on~$v$ are represented. It remains to prove that all edges incident on~$h$ are still represented after the swap. For this it suffices to observe that the only vertices that were in ancestor-descendant relation with~$h$ in~$\hat{T}$, but are not in ancestor-descendant relation with~$h$ in~$\widetilde{T}$, are those contained in subtrees of~$\hat{T}$ attached to the path~$P$ at vertices other than~$c_1, \ldots, c_t$. But by our choice of~$c_1, \ldots, c_t$, such subtrees contain no vertices of~$H$. By the same argumentation as in Claim~\ref{claim:subtrees:no:nh}, such subtrees contain no vertex of~$N_G(H)$ either. So the vertices to which ancestor-descendant relation is lost are not neighbors of~$h$ in~$G$, which implies that~$\widetilde{T}$ is indeed a valid decomposition of~$G$.

We finish the proof by applying the induction hypothesis. Since~$\widetilde{T}$ was obtained from~$\hat{T}$ by swapping the labels of~$v$ and~$h$ in the tree, while~$\depth(v,\hat{T}) > \depth(h,\hat{T})$ since~$h$ is an ancestor of~$v$ in~$\hat{T}$, it follows that~$\depth(v,\widetilde{T}) < \depth(v,\hat{T}) \leq \depth(v,T^*)\leq \depth(v,T)$. This implies that we may apply induction to~$G,H,v$, and the decomposition~$\widetilde{T}$, to conclude that there is a treedepth decomposition~$T'$ such that~$V(H) \subseteq T'_v$ and~$\height(T')$ is bounded by $$\max (\height(\widetilde{T}), \depth(v, \widetilde{T}) + \td(G[V(H)])) \leq \max(\height(T), \depth(v, T) + \td(G[V(H)])).$$ This concludes the proof.
\end{proof}

\subsection{Nearly clique-separated sets} \label{section:nearly:clique:isolated}
The purpose of this section is to introduce the following notion.

\begin{definition} \label{definition:nearly:clique:sep}
Let~$G$ be a graph, let~$S \subseteq V(G)$, and let~$\ell$ be an integer. The set~$S$ is \emph{$\ell$-nearly clique separated} if there is a set~$Q \subseteq N_G(S)$ of size at most~$\ell$ such that~$N_G(S) \setminus Q$ is a clique.
\end{definition}

Nearly clique separated sets are important because minimum treedepth modulators contain only few of their vertices, implying that the structure of a nearly clique separated subgraph does not change too much when a minimum modulator is removed.

\begin{lemma} \label{lem:boundingsolnintersection:updated}
Let~$G$ be a graph, let~$\ell$ be an integer, and let~$S\subseteq V(G)$ with~$\td(G[S])\leq \eta$. If~$S$ is $\ell$-nearly clique separated, then~$|Z \cap S| \leq \eta + \ell$ for any minimum treedepth-$\eta$ modulator~$Z$ of~$G$.
\end{lemma}
\begin{proof}
Assume for a contradiction that~$Z'$ is a minimum treedepth-$\eta$ modulator of~$G$ such that $|Z \cap S|> \eta + \ell$. Let~$Q \subseteq N_G(S)$ be a set of size at most~$\ell$ such that~$N_G(S) \setminus Q$ is a clique, which exists by Definition~\ref{definition:nearly:clique:sep}. Let~$K := N_G(S) \setminus Q$. Then~$G[K]$ is a clique, and therefore~$|K \setminus Z'| \leq \eta$ (otherwise the~$\eta + 1$ remaining vertices of the clique would cause the treedepth of~$G - Z'$ to exceed~$\eta$). We claim that $Z=(Z' \setminus S) \cup (K \setminus Z') \cup Q$ is a treedepth-$\eta$ modulator of $G$ of size less than~$|Z'|$.

The fact that the treedepth of $G-Z$ is at most~$\eta$ can be verified as follows. Observe that $\td(G[S])\leq \eta$ and~$N_G(S) \subseteq Z$, implying that no vertex of~$S$ is adjacent to a vertex outside of~$S$ in~$G - Z$. Hence all connected components of~$G-Z$ that contain a vertex of~$S$, have treedepth at most~$\eta$. On the other hand, for every connected component $H$ of $G-Z$ with $V(H)\cap S=\emptyset$ there exists a connected component~$H'$ of~$G-Z'$ such that $H \subseteq H'$ and thus, $\td(H)\leq \td(H')\leq \eta$. To see that~$Z$ is smaller than~$Z'$, observe that we remove~$|Z' \cap S| > \eta + \ell$ vertices from~$Z'$, while we add~$|K \setminus Z'| \leq \eta$ plus~$|Q| \leq \ell$ vertices. Hence~$Z$ is a smaller treedepth-$\eta$ modulator than~$Z'$, contradicting optimality of~$Z'$.
\end{proof}

\subsection{Safe Transformations} \label{section:treedepth:transformations}

In this section we analyze three different operations in a graph: removing vertices, removing edges, and adding edges. For each type of operation we give conditions under which the transformation provably does not change the answer to an instance~$(G,k)$ of \TreedepthEtaDeletion. We say that instances~$(G,k)$ and~$(G',k)$ are \emph{equivalent} if one is a \yes-instance if and only if the other is. When proving that two instances are equivalent, we frequently use the fact that if~$G'$ is a minor of~$G$ and~$(G,k)$ is a \yes-instance of \TreedepthEtaDeletion, then~$(G',k)$ is a \yes-instance as well. This follows from the fact that treedepth does not increase when taking minors, so that if~$\td(G - Z) \leq \eta$ we must have~$\td(G' - Z) \leq \eta$ since~$G' - Z$ is a minor of~$G - Z$.

\begin{figure}[t]
\begin{center}
\subfigure[Graph~$G$, set~$S \subseteq V(G)$.]{\label{fig:remove:simplicial:graph}
\image{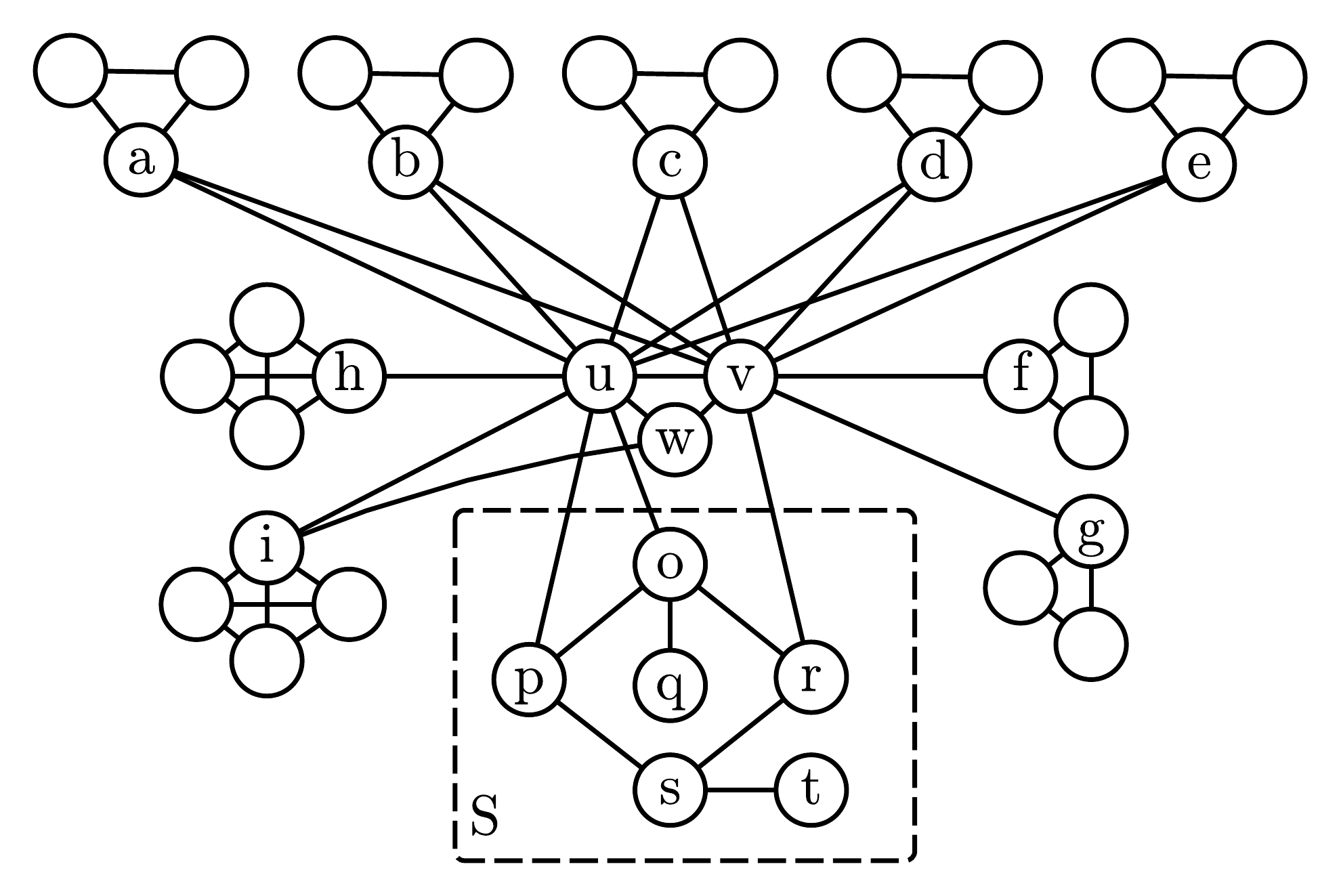}
}
\subfigure[Decompositions~$F$ and~$F'$.]{\label{fig:remove:simplicial:decompositions}
\image{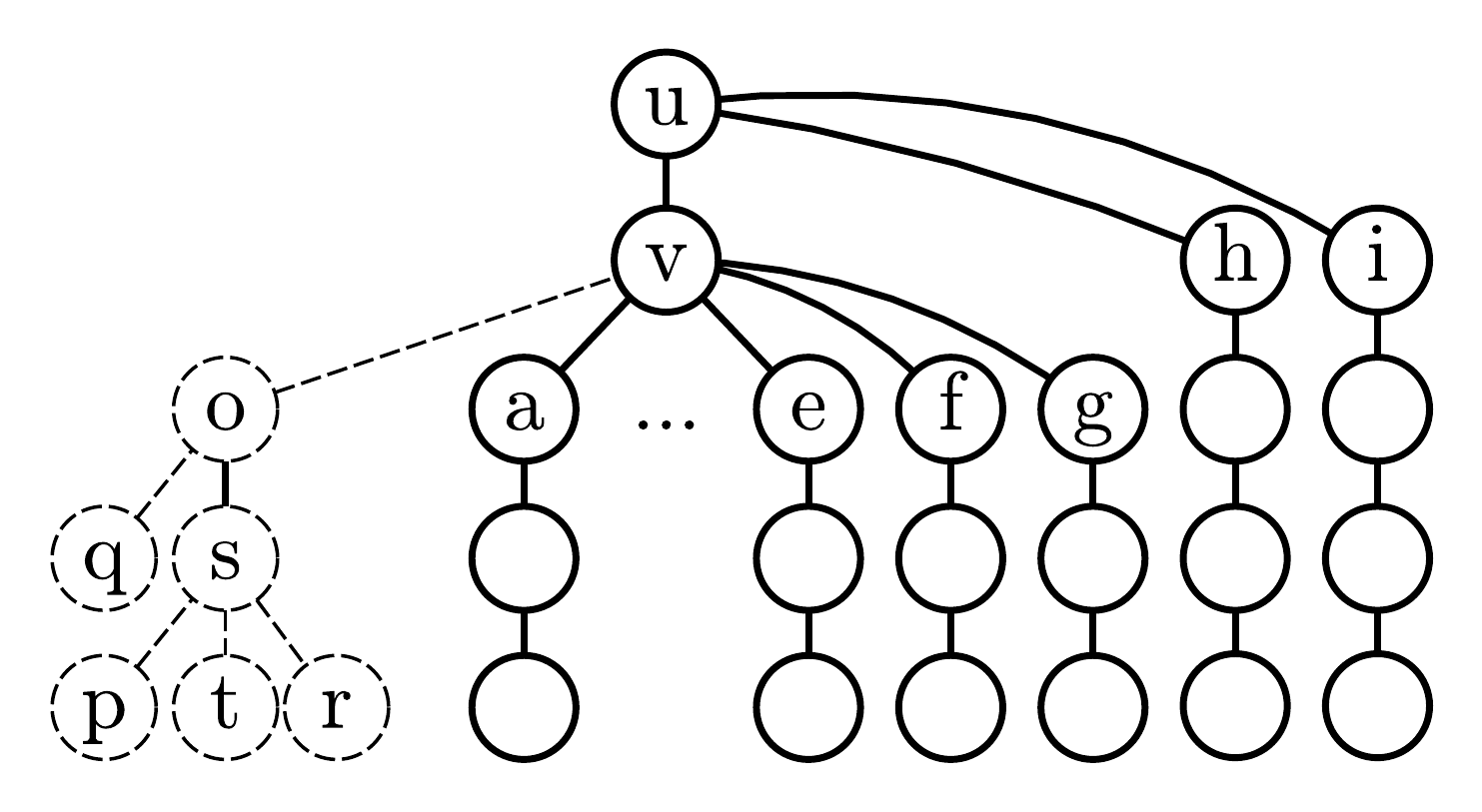}
}
\caption{Illustration for the proof of Lemma~\ref{lemma:remove:simplicial:component}. \ref{fig:remove:simplicial:graph} Graph~$G$ with vertex set~$S \subseteq V(G)$ whose neighborhood~$N_G(S) = \{u,v\}$ is a clique. This is a \yes-instance for treedepth-5 transversal with~$k=1$, since~$\{w\}$ is a solution. Lemma~\ref{lemma:remove:simplicial:component} is applicable with~$\ell=2$ by choosing~$X^v_1, \ldots, X^v_7$ to be the triangles containing~$a,\ldots,e,f,g$ respectively. The sets~$X^u_1,\ldots,X^u_7$ are the triangles containing~$a,\ldots,e$ and the 4-cliques containing~$h$ and~$i$, respectively; note that sets~$X^u_i$ may intersect sets~$X^v_j$, but that, e.g.,~$X^u_i$ and~$X^u_j$ must be disjoint. \ref{fig:remove:simplicial:decompositions} The solid lines show a treedepth-5 decomposition for the graph~$(G - S) - \{z\}$, where~$\{z\}$ is a solution to the reduced instance~$G - S$. The dotted lines show how to augment to a decomposition for the graph~$G - \{w\}$ by attaching a minimum-width decomposition for~$G[S]$ to the lowest neighbor of~$S$.}
\end{center}
\end{figure}

\begin{lemma}[Vertex removal] \label{lemma:remove:simplicial:component}
Let~$(G,k)$ be an instance of \TreedepthEtaDeletion and let~$\ell$ be an integer. Let~$S \subseteq V(G)$ be such that~$N_G(S)$ is a clique and~$\td(G[S]) \leq \eta$. For every~$v \in N_G(S)$, let~$X^v_{1}, \ldots, X^v_{\ell + \eta} \subseteq V(G)$ induce connected subgraphs of~$G$ such that:
\begin{enumerate}
	\item $\forall v \in N_G(S), \forall i \in [\ell + \eta]\colon \td(G[X^v_{i}]) \geq \td(G[S])$ and~$v \in N_G(X^v_{i})$,\label{prop:remove:simplicial:highdepth}
	\item $\forall v \in N_G(S)$, the sets~$X^v_{1}, \ldots, X^v_{\ell + \eta}$ are pairwise disjoint and disjoint from~$S$, and
	\item $G-S$~has a minimum treedepth-$\eta$ modulator containing~$\leq \ell$ vertices of~$\X$, 
\end{enumerate}
where~$\X := \bigcup _{v \in N_G(S)} \bigcup _{i \in [\ell + \eta]} X^v_{i}$. Then~$(G,k)$ is equivalent to the instance~$(G - S, k)$.
\end{lemma}
\begin{proof}
If~$(G,k)$ is a \yes-instance then~$(G - S, k)$ trivially is a \yes-instance as well. For the reverse direction, let~$Z$ be a minimum treedepth-$\eta$ modulator of~$G - S$ with~$|Z \cap \X| \leq \ell$, which exists by assumption and has size at most~$k$. Let~$F$ be a minimum-height treedepth decomposition of~$G - Z$, which has height at most~$\eta$. If~$N_G(S) \setminus Z = \emptyset$ then~$S$ forms an isolated component of treedepth at most~$\eta$ in the graph~$G - Z$, implying that~$\td(G - Z) = \max(\td((G - S) - Z), \td(G[S])) \leq \eta$. Assume then that~$N_G(S) \setminus Z$ is not empty. Since the set~$N_G(S)$ is a clique in~$G$, it follows that~$N_G(S) \setminus Z$ is a clique in~$G - Z$. Hence all vertices of~$N_G(S) \setminus Z$ appear on one root-to-leaf path in~$F$. Let~$v$ be the vertex of~$N_G(S) \setminus Z$ of greatest depth in~$F$.

\begin{claim} \label{claim:height:remove:simplicial}
We have~$\reach(v, F) \geq \td(G[S])$.
\end{claim}
\begin{claimproof}
Since the sets~$X^v_{1}, \ldots, X^v_{\ell + \eta} \subseteq \X$ are pairwise disjoint, each of the at most~$\ell$ vertices in~$Z \cap \X$ intersects at most one such subset. Hence there are at least~$\eta$ such subsets~$X^v_{j_1}, \ldots, X^v_{j_{\eta}}$ that are not intersected by~$Z$. Each~$X^v_{j_t}$ therefore induces a connected subgraph of~$(G - S) - Z$ of treedepth at least~$\td(G[S])$ by~(\ref{prop:remove:simplicial:highdepth}) that is adjacent in~$(G - S) - Z$ to~$v$. Now apply Lemma~\ref{lem:neighheght} to the connected component of~$(G - S) - Z$ containing~$v$ and the tree~$T$ in~$F$ representing that component, using~$X^v_{j_1}, \ldots, X^v_{j_{\eta}}$ as connected subgraphs of treedepth at least~$\td(G[S])$. The lemma shows that if~$\reach(v, T) < d$, then~$\height(T) \geq \eta + 1$, contradicting the fact that~$\height(T) \leq \height(F) \leq \eta$. Hence~$\reach(v,T) \geq d$. Since~$F$ is at least as high as~$T$, this implies~$\reach(v,F) \geq \td(G[S])$.
\end{claimproof}

Starting from the decomposition~$F$ of~$(G - S) - Z$, we now obtain a decomposition~$F'$ of~$G - S$ as follows. Take a minimum-height decomposition forest of~$G[S]$, add it to~$F$ and connect the roots of all vertices in the decomposition forest to~$v$. Since~$N_G(S) \setminus Z$ appears on the path from~$v$ to the root of its tree, this results in a valid treedepth decomposition of~$G - S$. Since~$\reach(v,F) \geq \td(G[S])$, the height of~$F'$ equals the height of~$F$. Hence~$G - Z$ has treedepth at most~$\eta$, showing that~$(G,k)$ is a \yes-instance.
\end{proof}

Lemma~\ref{lemma:remove:simplicial:component} was inspired by earlier work~\cite[Rule 6]{BodlaenderJK12a} on \textsc{Pathwidth}. The next lemma concerns edge addition.

\begin{lemma}[Edge addition] \label{lemma:add:edge}
Let~$(G,k)$ be an instance of \TreedepthEtaDeletion and let~$\ell$ be an integer. Let~$X \subseteq V(G)$ and let~$\{u,v\} \in \binom{V(G)}{2} \setminus E(G)$. If the following conditions hold:
\begin{enumerate}
	\item the graph~$G[X \cup \{u,v\}]$ contains at least~$\ell + \eta$ internally vertex-disjoint paths between~$u$ and~$v$, and
	\item $G$ has a minimum treedepth-$\eta$ modulator containing~$\leq \ell$ vertices of~$X$,
\end{enumerate}
then~$(G,k)$ is equivalent to the instance~$(G + \{u,v\}, k)$ obtained by adding the edge~$\{u,v\}$.
\end{lemma}
\begin{proof}
If~$(G + \{u,v\}, k)$ is a \yes-instance then~$(G,k)$ is as well. In the other direction, suppose that~$(G,k)$ is a \yes-instance and let~$Z$ be a minimum treedepth-$\eta$ modulator of~$G$ with~$|Z \cap X| \leq \ell$, which exists by assumption. Let~$F$ be a minimum-height treedepth decomposition of~$G - Z$. If~$Z \cap \{u,v\} \neq \emptyset$ then~$G - Z = (G + \{u,v\}) - Z$ and so~$Z$ is a solution for~$G + \{u,v\}$, proving it to be a \yes-instance.

Assume then that~$Z \cap \{u,v\} = \emptyset$. Since there are~$\ell + \eta$ internally vertex-disjoint paths between~$u$ and~$v$ in~$G[X \cup \{u,v\}]$, while~$Z$ intersects at most~$\ell$ of them as~$|Z \cap X| \leq \ell$, it follows that there are~$\eta$ internally vertex-disjoint paths between~$u$ and~$v$ in~$G - Z$. By Lemma~\ref{lem:bgflw}, these~$\eta$ paths prove that if~$u$ and~$v$ are not in ancestor-descendant relation in~$F$ then~$\height(F) > \eta$, a contradiction. So~$u$ and~$v$ are in ancestor-descendant relation which shows that~$F$ is also a valid treedepth decomposition of~$(G + \{u,v\}) - Z$, proving~$(G + \{u,v\}, k)$ to be a \yes-instance.
\end{proof}

Finally, we consider edge removal.

\begin{lemma}[Edge removal] \label{lemma:remove:edges}
Let~$(G,k)$ be an instance of \TreedepthEtaDeletion and let~$\ell$ be an integer. Let~$S \subseteq V(G)$ and let~$v \in V(G) \setminus S$ such that~$N_G(S) \subseteq N_G[v]$. Let~$X_1, \ldots, X_{\ell + \eta} \subseteq V(G)$ be connected subgraphs of~$G$ such that:
\begin{enumerate}
	\item $\forall i \in [\ell + \eta]\colon \td(G[X_{i}]) \geq \td(G[S])$ and~$v \in N_G(X_{i})$,
	\item the sets~$X_{1}, \ldots, X_{\ell + \eta}$ are pairwise disjoint and disjoint from~$S$, and
	\item any graph obtained from~$G$ by removing edges between~$v$ and~$S$ has a minimum treedepth-$\eta$ modulator containing~$\leq \ell$ vertices of~$\X$,\label{prop:remove:edges:smallint}
\end{enumerate}
where~$\X := \bigcup _{i \in [\ell + \eta]} X_i$. Then~$(G,k)$ is equivalent to the instance~$(G', k)$, where~$G'$ is obtained from~$G$ by removing all edges between~$v$ and~$S$.
\end{lemma}
\begin{proof}
If~$(G,k)$ is a \yes-instance then its minor~$(G',k)$ is as well. In the other direction, assume that~$(G',k)$ is a \yes-instance. Assume additionally that~$(G,k)$ is a \no-instance; we shall argue for a contradiction. Consider the set of edges~$Y \subseteq E(G)$ that were removed from~$G$ to obtain~$G' = G - Y$, and observe that all edges in~$Y$ are incident on~$v$. Let~$Y' \subseteq Y$ be a \emph{minimal} set such that~$G^* := G - Y'$ is a \yes-instance. Clearly~$Y'$ is not empty, as~$G$ is a \no-instance. Let~$\{u,v\}$ be an arbitrary edge of~$Y'$, which must have~$v$ as an endpoint. By minimality of~$Y'$ we know that~$G^* + \{u,v\} = G - (Y' \setminus\{u,v\})$ is a \no-instance. We will derive a contradiction by proving that~$G^* + \{u,v\}$ is actually a \yes-instance. By (\ref{prop:remove:edges:smallint}), graph~$G^*$ has a minimum treedepth-$\eta$ modulator~$Z$ containing at most~$\ell$ vertices of~$\X$. Since~$(G^*, k)$ is a \yes-instance, the size of~$Z$ is at most~$k$. Let~$F$ be a minimum-height, nice treedepth decomposition of~$G^* - Z$, of height at most~$\eta$. If~$\{u,v\} \cap Z \neq \emptyset$ then the graphs~$G^* - Z$ and~$(G^* + \{u,v\}) - Z$ are identical, so~$\td(G^* + \{u,v\}) - Z \leq \eta$, proving that~$G^* + \{u,v\}$ is a \yes-instance. In the remainder we consider the case that~$\{u,v\} \cap Z = \emptyset$.

\begin{claim} \label{claim:height:remove:edges}
We have~$\reach(v, F) \geq \td(G[S])$.
\end{claim}
\begin{claimproof}
The proof is similar to that of Claim~\ref{claim:height:remove:simplicial}. Since~$X_1, \ldots, X_{\ell + \eta}$ are pairwise disjoint subsets of~$\X$, the set~$Z$ intersects at most~$\ell$ of them. Hence at least~$\eta$ of them, say~$X_{j_1}, \ldots, X_{j_\eta}$, are disjoint from~$Z$ and therefore induce connected subgraphs of~$G^* - Z$ of treedepth at least~$\td(G[S])$. Since these sets are disjoint from~$S$ and they were adjacent to~$v$ in~$G$, we have not removed their edges to~$v$ when constructing~$G^*$ and therefore~$v \in N_{G^* - Z}(X_{j_t})$ for all~$t \in [\eta]$. Applying Lemma~\ref{lem:neighheght} to the connected component of~$G^* - Z$ containing~$v$ and the tree~$T$ in~$F$ representing that component, we find that if~$\reach(v, T) < d$, then~$\height(T) \geq \eta + 1$, a contradiction. Hence~$\reach(v,F) \geq \reach(v,T) \geq d$.
\end{claimproof}

We use the lower bound on~$\reach(v,F)$ in the following arguments. Let~$S_u$ be the connected component of~$G^*[S] - Z$ that contains~$u$, the other endpoint of~$\{u,v\}$. We first deal with an easy case. 

\begin{claim}
If~$N_{G^* - Z}(S_u) = \emptyset$, then~$(G^* + \{u,v\}, k)$ is a \yes-instance.
\end{claim}
\begin{claimproof}
If~$N_{G^* - Z}(S_u) = \emptyset$, then~$S_u$ forms a connected component of the graph~$G^* - Z$, since this connected set has no neighbors. From the definition of a nice treedepth decomposition, there is a single tree~$T_u$ in~$F$ whose vertices are~$S_u$. Since~$S_u$ is a subgraph of~$G[S]$ we have~$\td(G[S_u]) \leq \td(G[S])$. Now obtain a treedepth decomposition~$F'$ of~$(G^* + \{u,v\}) - Z$ as follows. Remove the tree~$T_u$ from~$F$, let~$T'_u$ be a minimum-height treedepth decomposition of~$G[S_u]$, add this tree to~$F$ and make the root of~$T'_u$ a child of~$v$. Since the height of~$T'_u$ is at most~$\td(G[S])$, while~$\reach(v,F)$ is at least~$\td(G[S])$, it follows that~$F'$ is not higher than~$F$. Since~$u \in S_u$ is a descendant of~$v$ in~$F'$, it follows that~$F'$ represents the edge~$\{u,v\}$ and is therefore a valid treedepth decomposition of~$(G^* + \{u,v\}) - Z$ of height at most~$\eta$. Hence~$Z$ is a treedepth-$\eta$ modulator of~$G' + \{u,v\}$ of size at most~$k$, showing that~$(G^* + \{u,v\}, k)$ is a \yes-instance.
\end{claimproof}

The claim shows that if~$N_{G^* - Z}(S_u) = \emptyset$ then~$(G^* + \{u,v\},k)$ is a \yes-instance, a contradiction to our starting assumption. Hence in the remainder it suffices to deal with the case that~$N_{G^* - Z}(S_u) \neq \emptyset$. Let~$x$ be a vertex in~$N_{G^* - Z}(S_u)$ and let~$y$ be a neighbor of~$x$ in~$S_u$. Since~$x \in N_{G^* - Z}(S_u)$ and~$S_u$ is a connected component of~$G^*[S] - Z$, it follows that~$x \not \in S$. Hence~$x \in N_{G^*}(S)$.

\begin{claim}
Vertices~$u$ and~$v$ belong to the same connected component of~$G^* - Z$.
\end{claim}
\begin{claimproof}
Observe that there exists a path from~$u$ to~$y$ in~$S_u$, since~$S_u$ is a connected component of~$G^* - Z$. If~$x = v$ then combining this path from~$u$ to~$y$ in~$S_u$ with an edge from~$y$ to~$x=v$, we obtain a $uv$ path in~$G^* - Z$, proving the claim. Assume then that~$x \neq v$. Since~$x \in N_{G^*}(S)$ and~$G^* \subseteq G$, it follows that~$x \in N_G(S)$. Since we did not remove edges between~$v$ and vertices outside~$S$ when forming~$G^*$, while~$N_G(S) \subseteq N_G[v]$, it follows that~$x \in N_G(v)$. Now we obtain a path from~$u$ to~$v$ in~$G^* - Z$ as follows: start with the path from~$u$ to~$y$ in~$S_u$, follow the edge to~$x$, and finally follow the edge to~$v$.
\end{claimproof}

Using the claim we can finish the proof. Let~$G_{uv}$ be the connected component of~$G^* - Z$ containing~$u$ and~$v$. Let~$T$ be the tree in~$F$ representing~$G_{uv}$. Since~$N_G(S) \subseteq N_G[v]$, while~$S_u$ is a connected component of~$G^* - Z$, it follows that~$N_{G^* - Z}(S_u) \subseteq N_{G^* - Z}[v]$. We may therefore apply Lemma~\ref{lemma:switch} to the connected graph~$G_{uv}$, the vertex~$v$, and the connected subgraph~$H := S_u$ of~$G_{uv}$, along with the decomposition~$T$ of~$G_{uv}$. The lemma guarantees that there is a decomposition~$T'$ of~$G_{uv}$ such that all vertices of~$S_u$ are in the subtree of~$T'$ rooted at~$v$, and~$\height(T') \leq \max (\height(T), \depth(v,T) + \td(G_{uv}[S_u]) )$. It is easy to see that if we replace~$T$ by~$T'$ in the decomposition~$F$, we obtain a valid treedepth decomposition~$F'$ of~$G^* - Z$. Since~$u$ is in the subtree rooted at~$v$, the decomposition represents the edge~$\{u,v\}$ and is therefore also a decomposition of~$(G^* + \{u,v\}) - Z$. It remains to bound the height of~$F$, for which it suffices to bound the height of~$T'$.

By Claim~\ref{claim:height:remove:edges} we have~$\reach(v,F) \geq \td(G[S])$. Using the definition of reach this implies that~$\height(F) \geq \depth(v,F) + \td(G[S])$. Since~$\td(G_{uv}[S_u]) \leq \td(G[S])$ this implies that~$\height(T') \leq \max(\height(T), \height(F)) \leq \height(F)$ by the expression above. Hence~$T'$ does not increase the height of~$F'$ beyond~$\eta$, showing that~$F'$ is a treedepth decomposition of~$G^* + \{u,v\}$ of height at most~$\eta$. Hence~$Z$ is a solution for~$G^* + \{u,v\}$ of size at most~$k$, proving that~$(G^* + \{u,v\}, k)$ is a \yes-instance. As this contradicts our starting assumption, this concludes the proof of Lemma~\ref{lemma:remove:edges}.
\end{proof}
\section{Uniformly polynomial kernelization for Treedepth-$\eta$ Deletion}
In this section we develop the kernelization for \TreedepthEtaDeletion. As described in the introduction, the two main ingredients are a decomposition algorithm (Section~\ref{section:decompose}) and a reduction algorithm (Section~\ref{section:reduction:algorithm}) that will be applied to each piece of the decomposition. These will be combined into the final kernelization algorithm in Section~\ref{section:final:kernel}.

\subsection{Structural decomposition of the input graph} \label{section:decompose}

We present the algorithm that decomposes an instance~$(G,k)$ into a small number of pieces that each have a constant-size intersection with any minimum solution. The procedure is given as Algorithm 1. In the following lemma we analyze its behavior. Let us point out that the sets~$S$ and~$Y$ computed by the algorithm decompose the graph into $\eta$-nearly clique separated components~$C$ of~$G - (S \cup Y)$: the neighborhood of each component~$C$ in~$S$ is a clique, and its neighborhood in the rest of the graph is contained on one root-to-leaf path in the decomposition~$F$ and therefore has size at most~$\eta$. The intersection size of minimum treedepth-$\eta$ modulators with such components is therefore at most~$2\eta$ by Lemma~\ref{lem:boundingsolnintersection:updated}.

\begin{algorithm}[t]
\caption{Decompose(Graph~$G$, $\eta \in \mathbb{N}$, $k \in \mathbb{N}$)} \label{alg:decompose}
\begin{algorithmic}[1]

\vspace{0.2cm}

\WHILE{$\exists$ distinct~$p,q \in G$ such that~$\{p,q\} \not \in E(G)$ and~$\lambda_G(p,q) \geq k + \eta$}
	\STATE Add the edge~$\{p,q\}$ to~$G$\label{line:decompose:addedge}
\ENDWHILE

\COMMENT{All non-adjacent pairs~$\{p,q\}$ at this point satisfy~$\lambda_G(p,q) < k + \eta$}

\STATE Apply Lemma~\ref{lemma:tdmodapprox} on the current graph to compute an approximate treedepth-$\eta$ modulator~$S$\label{line:decompose:approximate}

\IF{$|S| > 2^\eta \cdot k$}
	\STATE Report that the original input graph does not have a treedepth-$\eta$ modulator of size~$\leq k$
\ENDIF

\STATE Initialize~$Y_0$ and~$Y_1$ as empty vertex sets

\FOREACH{$\{p, q\} \in \binom{S}{2} \setminus E(G)$}
	\STATE Let~$Y_{p,q} \subseteq V(G) \setminus \{p,q\}$ be a minimum $pq$-separator
	\COMMENT{Menger's theorem: $|Y_{p,q}| < k + \eta$}
	\STATE Add~$Y_{p,q}$ to~$Y_0$\label{line:add:separator}
\ENDFOR

\STATE Compute a minimum-height nice treedepth decomposition~$F$ of~$G - S$ using Lemma~\ref{lemma:opt:nice:decomposition}\label{line:decompose:computef}

\FOREACH{$v \in Y_0$}
	\STATE Add the proper ancestors~$\anc_F(v)$ of~$v$ in~$F$ to~$Y_1$ 
	\COMMENT{Since~$F$ has height~$\leq \eta$, $|\anc_F(v)| < \eta$}
\ENDFOR

\STATE Let~$Y$ be~$Y_0 \cup Y_1$
\STATE Define~$\T := \{u \in V(F) \setminus Y \mid u \mbox{ is a root or } \pi(u) \in Y\}$

\WHILE{there is a node~$u_0$ in~$\T$ such that: 
\begin{enumerate}
	\item $G[N_G(F_{u_0})]$ is a clique, and 
	\item for each~$v \in N_G(F_{u_0})$ there are distinct nodes~$u^v_1, \ldots, u^v_{\eta + k} \in \T \setminus \{u\}$ such that:$$\forall i \in [\eta + k] \colon v \in N_G(F_{u_i}) \wedge \td(G[F_{u_i}]) \geq \td(G[F_{u_0}])$$
\end{enumerate}
}\label{line:decompose:while:removecomponents}
	\STATE Remove the vertices of~$F_{u_0}$ from~$G$ and~$F$ and remove~$u$ from~$\T$\label{line:decompose:removecomponents}
	\COMMENT{For~$v \neq v' \in N_G(F_{u_0})$ we may have~$u^v_i = u^{v'}_j$}
\ENDWHILE

\STATE Output the updated graph and the decomposition~$F$, the modulator~$S$, and the separator~$Y$
\end{algorithmic}
\end{algorithm}

\begin{figure}[t]
\begin{center}
\subfigure[Graph~$G'$, modulator~$S$.]{\label{fig:decomposition1}
\image{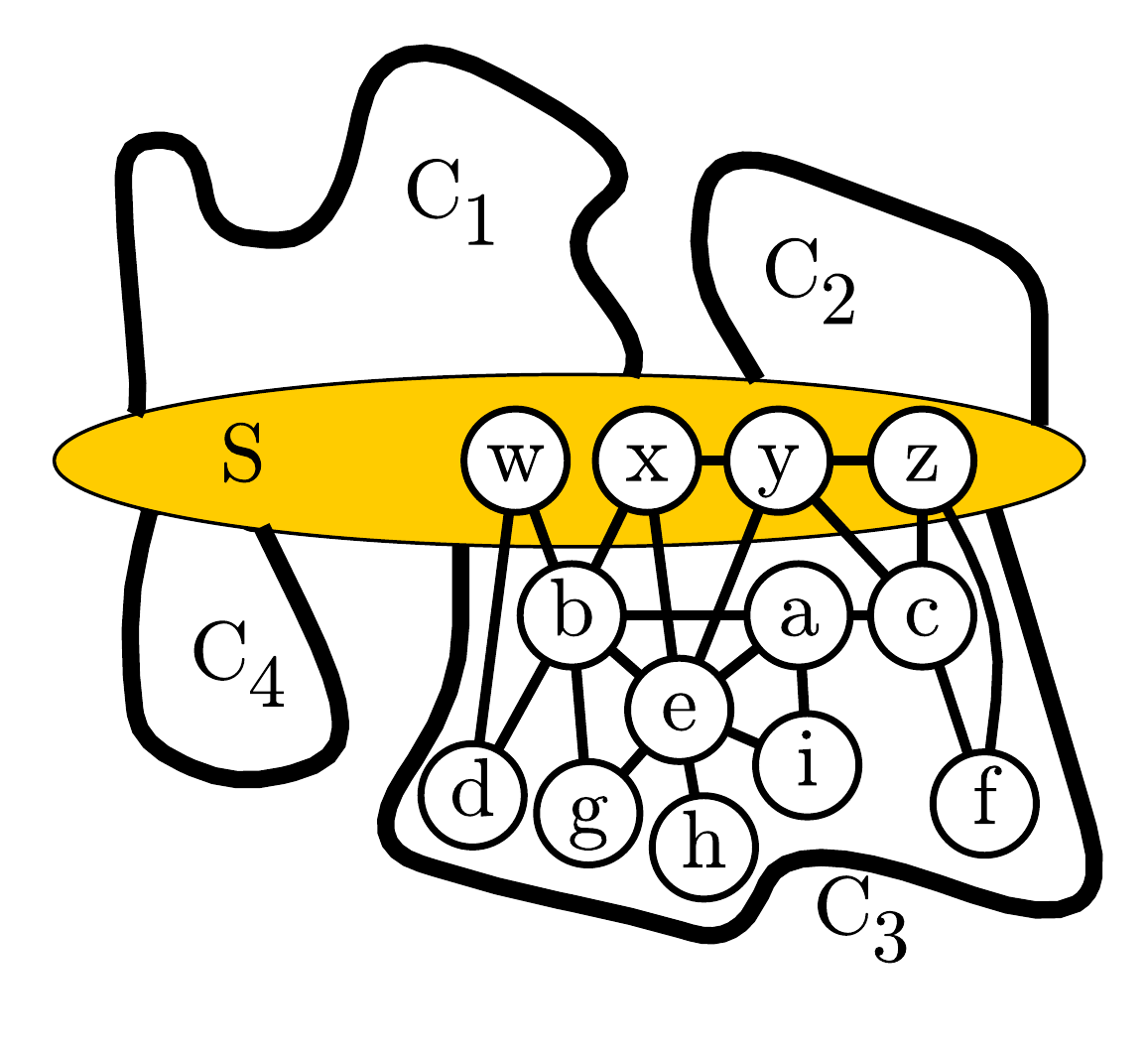}
}
\subfigure[Decomposition of~$G'-S$.]{\label{fig:decomposition2}
\image{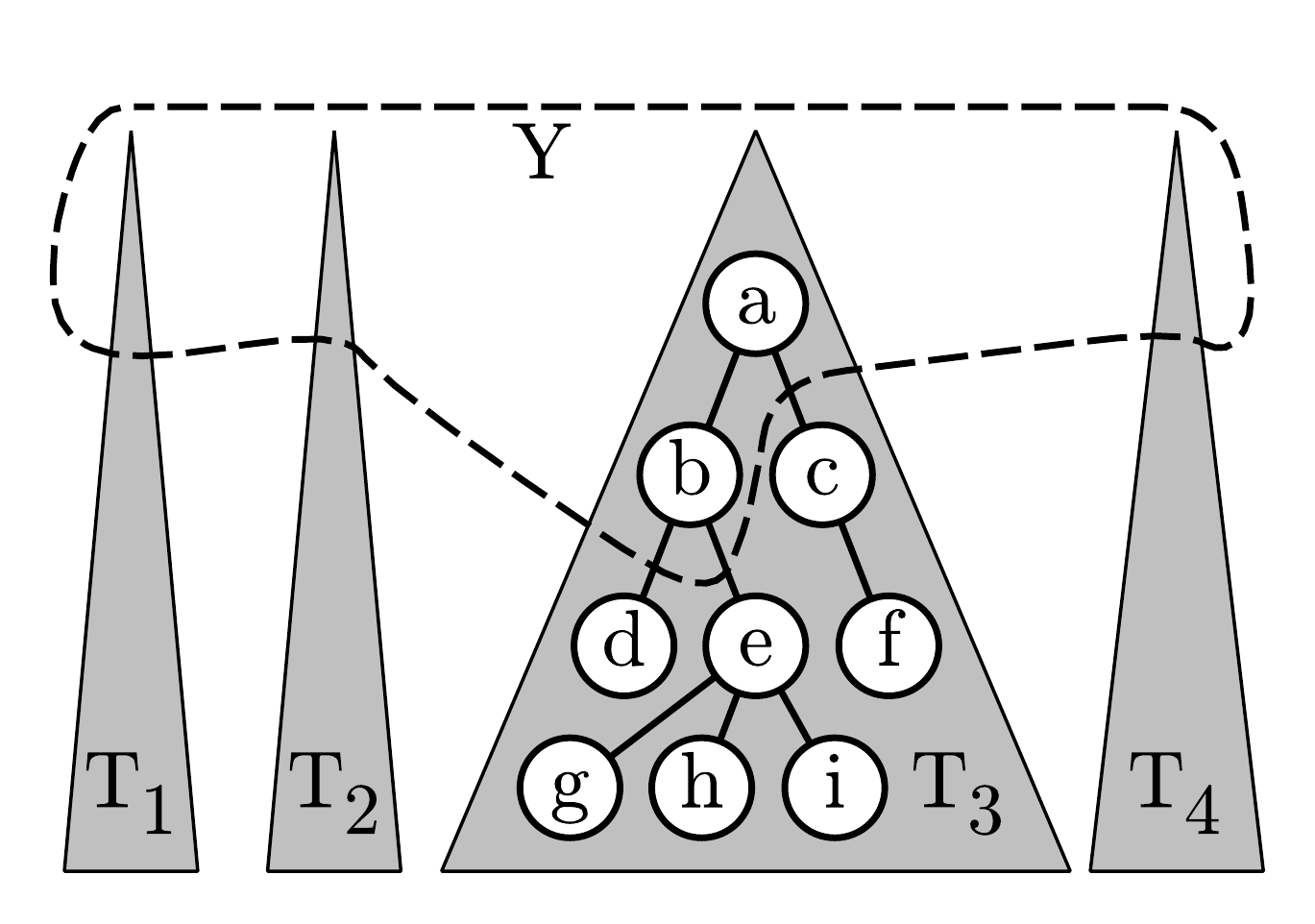}
}
\caption{Schematic illustration of an instance that has been decomposed using Lemma~\ref{lemma:decomposition:algorithm}. \ref{fig:decomposition1} The resulting graph~$G'$ and the suboptimal treedepth-$4$ modulator~$S$ in~$G'$ used when decomposing. Graph~$G'-S$ has four connected components, of which the third is drawn in detail. \ref{fig:decomposition2} Illustration of the treedepth-$4$ decomposition~$F'$ of~$G'-S$. The forest~$F'$ contains four decomposition trees~$T_1, \ldots, T_4$, one for each component of~$G'-S$. By the properties of a treedepth decomposition, for any vertex~$v \in V(G') \setminus S$, each neighbor~$u \in N_{G'}(v)$ is an ancestor of~$v$ in~$F'$, descendant of~$v$ in~$F'$, or contained in~$S$. Lemma~\ref{lemma:decomposition:algorithm} ensures that for each connected component~$C$ of~$G - (S \cup Y)$, the set~$N_G(C) \cap S$ is a clique. This is illustrated for the connected component consisting of~$\{e,g,h,i\}$, whose neighbors among~$S$ are~$\{x,y\}$, a 2-clique. As the set~$Y$ is closed under taking ancestors, it consists of the top parts of decomposition trees in~$F'$.}
\end{center}
\end{figure}

\begin{lemma} \label{lemma:decomposition:algorithm}
Let $(G,k)$ be an instance of \probtdeta. Then in polynomial time we can either conclude that $(G,k)$ is a \noinstance, or find a graph~$G'$ with~$V(G') \subseteq V(G)$, a treedepth-$\eta$ modulator~$S$ of~$G'$, a treedepth decomposition~$F'$ of~$G' - S$ of height at most~$\eta$, and a set~$Y \subseteq V(G') \setminus S$ satisfying the following properties.
\begin{enumerate}[(\arabic*)]
	\item\label{dec:equivalent}$(G,k)$ is equivalent to~$(G',k)$.
	\item\label{dec:small:s} $|S| \leq 2^\eta \cdot k$.
	\item\label{dec:small:y} $|Y| \leq \eta (2^\eta \cdot k)^2 \cdot (k + \eta)$.
	\item\label{dec:connected:subgraphs} For every~$u \in V(F') \setminus Y$ the graph~$G'[F'_u]$ is connected.
	\item\label{dec:component:structure} Let~$\T := \{u \in F' - Y \mid u \mbox{ is a root or } \pi(u) \in Y\}$. The vertex sets of the connected components of~$G' - (S \cup Y)$ are exactly the vertex sets of the subtrees of~$F'$ rooted at members of~$\T$.
	\item\label{dec:component:clique} For every connected component~$C$ of~$G' - (S \cup Y)$, the set~$N_{G'}(C) \cap S$ is a clique.
	\item\label{dec:component:number} The number of connected components of~$G' - (S \cup Y)$ is at most $$(|S| + |Y| + |S|^2 + |S| \cdot |Y| + \eta \cdot |Y|) \cdot (\eta + k).$$
\end{enumerate}
\end{lemma}

\begin{proof}
We prove that the Decompose algorithm has the desired properties. Let us go through the algorithm line by line to analyze its effect. Along the way we will establish that the requirements from the lemma are satisfied. Consider an execution of Decompose$(G, \eta, k)$ and let us denote by~$G^0$ the state of the graph before the execution starts.

Let us first consider the effect of the edges that are added in Line~\ref{line:decompose:addedge}.

\begin{claim} \label{claim:decompose:equiv:addedges}
If the algorithm transforms~$G$ into~$G + \{p,q\}$ by adding an edge in Line~\ref{line:decompose:addedge}, then the \TreedepthEtaDeletion instance~$(G,k)$ is equivalent to instance~$(G + \{p,q\}, k)$.
\end{claim}
\begin{claimproof}
If~$(G + \{p,q\}, k)$ is a \yes-instance then clearly its minor~$(G,k)$ is as well. So assume that~$(G, k)$ is a \yes-instance, implying that minimum treedepth-$\eta$ modulators in~$G$ have size at most~$k$. Define~$\ell := k$ and let~$X := V(G)$. Since the algorithm ensures that~$\lambda_G(p,q) \geq k + \eta$, the graph~$G[X \cup \{p,q\}] = G$ contains at least~$k + \eta$ internally vertex-disjoint $pq$-paths. Since a minimum solution in~$G$ contains at most~$k$ vertices, trivially any minimum treedepth-$\eta$ modulator intersects~$X$ in at most~$k = \ell$ vertices. But then these choices of~$X$ and~$\ell$ satisfy the conditions of Lemma~\ref{lemma:add:edge}, which proves that~$(G,k)$ is equivalent to~$(G+\{p,q\}, k)$.
\end{claimproof}

The claim shows that every edge addition preserves the answer to the instance. Hence for the graph~$G$ obtained after finishing the first \textbf{while}-loop, instance~$(G,k)$ is equivalent to~$(G^0,k)$. Once the \textbf{while}-loop terminates, for each non-adjacent vertex pair~$\{p,q\}$ that remains, the number of internally disjoint $pq$-paths must be less than~$k + \eta$. Let us consider the set~$S$ that is computed, and denote by~$G^1$ the status of the graph in Line~\ref{line:decompose:approximate}. By Lemma~\ref{lemma:tdmodapprox} the treedepth of~$G^1 - S$ is at most~$\eta$. If~$|S| > 2^\eta \cdot k$, since Lemma~\ref{lemma:tdmodapprox} guarantees a factor~$2^\eta$-approximation, the minimum solution size for~$(G^1,k)$ exceeds~$k$. By the equivalence of the instance to~$(G^0,k)$ the algorithm is therefore correct if it outputs that the original input~$G^0$ does not have a size-$k$ solution. If a set~$S$ is returned it must therefore satisfy~\ref{dec:small:s}. We continue by analyzing the computed set~$Y$.

\begin{claim}
$|Y_0| \leq |S|^2 \cdot (k + \eta) \leq (2^\eta \cdot k)^2 \cdot (k + \eta)$.
\end{claim}
\begin{claimproof}
Since~$|S| \leq 2^\eta \cdot k$, there are at most~$(2^\eta \cdot k)^2$ pairs of nonadjacent vertices among~$S$. For every nonadjacent pair of vertices~$\{p,q\}$, the termination condition of the \textbf{while}-loop ensures that~$\lambda_{G^1}(p,q) \leq \eta + k$. By Menger's Theorem the value~$\lambda_{G^1}(p,q)$ equals the minimum size of a vertex $pq$-separator that avoids~$p$ and~$q$; such a set can be computed efficiently (cf.~\cite[Chapter 9]{Schrijver03}). Hence~$|Y_{pq}| < k + \eta$ for all considered pairs. The claim follows.
\end{claimproof}

\begin{claim}
$|Y| \leq \eta \cdot (2^\eta \cdot k)^2 \cdot (k + \eta)$.
\end{claim}
\begin{claimproof}
Since~$Y_1$ contains the proper ancestors of every member of~$Y_0$, while every vertex has at most~$\eta - 1$ proper ancestors in a treedepth decomposition of height at most~$\eta$, it follows that~$|Y_1| \leq (\eta-1) |Y_0|$. Using the previous claim, we find that~$Y = Y_0 \cup Y_1$ has size at most~$\eta \cdot (2^\eta \cdot k)^2 \cdot (k + \eta)$.
\end{claimproof}

The claim shows that requirement~\ref{dec:small:y} is satisfied. Before we analyze last \textbf{while}-loop, we consider the structure of the computed set~$\T$ of \emph{topmost} vertices in the forest that do not belong to~$Y$. 

\begin{claim} \label{claim:distinct:tops:distinct:subtrees}
If~$u, u'$ are distinct vertices in~$\T$ then the subtrees~$F_u, F_{u'}$ are disjoint.
\end{claim}
\begin{claimproof}
Assume for a contradiction that the claim is false. By symmetry, we may assume that~$u'$ is in the subtree of~$F$ rooted at~$u$. Then~$u'$ is not the root of a tree. By definition of the set~$\T$ this implies that~$\pi(u')$ is a vertex in~$Y$. But since we added the ancestor of every vertex in~$Y_0$ to~$Y_1$, this implies that the ancestor~$u$ of~$u'$ must be contained in~$Y$. By definition of~$\T$ this contradicts that~$u \in \T$.
\end{claimproof}

We observe that, by updating the set~$\T$ in Line~\ref{line:decompose:removecomponents}, the algorithm ensures that at any point of its execution of the last \textbf{while}-loop, even though the graph might have changed after the point that~$\T$ was defined and computed, the set~$\T$ still satisfies that definition.

Using Claim~\ref{claim:distinct:tops:distinct:subtrees} we analyze what happens in the \textbf{while}-loop of Line~\ref{line:decompose:while:removecomponents}. In Line~\ref{line:decompose:removecomponents} we consider the vertices contained in the subtree of~$F$ rooted at a node~$u_0$ that satisfies the conditions of the \textbf{while}-loop, and we remove them from the graph.

\begin{claim} \label{claim:decompose:equiv:removesubtrees}
If the algorithm transforms~$G$ into~$G - F_{u_0}$ in Line~\ref{line:decompose:removecomponents}, then the \TreedepthEtaDeletion instance~$(G,k)$ is equivalent to instance~$(G - F_{u_0}, k)$.
\end{claim}
\begin{claimproof}
Define~$A$ as the vertices in the subtree of~$F$ rooted at~$u_0$. If~$(G,k)$ is a \yes-instance then its minor~$(G - A, k)$ is as well. For the reverse direction, suppose that~$(G - A, k)$ is a \yes-instance. We aim at applying Lemma~\ref{lemma:remove:simplicial:component}. By the preconditions to the loop, we know that~$N_G(A)$ is a clique in~$G$. Define~$\ell := k$. Then, clearly, for any~$\X \subseteq V(G)$ the graph~$G - A$ has a minimum treedepth-$\eta$ modulator intersecting~$\X$ in at most~$\ell$ vertices, showing that the third requirement of Lemma~\ref{lemma:remove:simplicial:component} is satisfied. Let us verify the first two requirements are satisfied as well. For each~$v \in N_G(A)$ and~$i \in [\eta + k]$, define~$X^v_i$ as the vertices in the subtree of~$F$ rooted at the node~$u^v_i$ identified in the algorithm. Then the test in the algorithm ensures the first condition of Lemma~\ref{lemma:remove:simplicial:component} is satisfied. The fact that, for each~$v \in N_G(A)$, the vertex sets~$X^v_i$ are pairwise disjoint and disjoint from~$A$, follows from Claim~\ref{claim:distinct:tops:distinct:subtrees}. It remains to check that all sets~$X^v_i$ induce connected subgraphs of~$G$. If the treedepth decomposition~$F$ is still nice when the statement is executed, then this follows from the definition of a nice treedepth decomposition. While earlier removals may have caused~$F$ to no longer be a nice decomposition, since the sets~$X^v_i$ correspond to subtrees of the nice treedepth forest originally computed in Line~\ref{line:decompose:computef}, and other iterations of the loop do not affect the graphs they induce, all sets~$X^v_i$ indeed induce connected subgraphs of~$G$ at the time the statement is executed. Hence all requirements are met and Lemma~\ref{lemma:remove:simplicial:component} implies the claim.
\end{claimproof}

The combination of Claims~\ref{claim:decompose:equiv:addedges} and~\ref{claim:decompose:equiv:removesubtrees} proves that for the state~$G'$ of the graph upon termination, instance~$(G^0,k)$ is equivalent to~$(G',k)$. Hence~\ref{dec:equivalent} holds.

\begin{claim}
The \textbf{while}-loop of Line~\ref{line:decompose:while:removecomponents} can be evaluated in polynomial time for every fixed~$\eta$.
\end{claim}
\begin{claimproof}
As each iteration removes a vertex, the number of iterations is bounded by the order of the input graph. Let us prove that each iteration can be done in polynomial time. To test the loop condition, it suffices to do the following. For each~$u \in \T$ we consider the subtree~$F_u$ rooted at~$F$ and compute a minimum-height treedepth decomposition of~$G[F_u]$. Since~$F_u \subseteq V(G) \setminus S$ the treedepth is at most~$\eta$, this can be done in polynomial time for fixed~$\eta$ by Lemma~\ref{lemma:opt:nice:decomposition}. For each possible choice of~$u_0$ we can then test whether the conditions hold for~$u_0$ by checking whether, for each~$v \in N_G(F_u)$, there are sufficiently many components also adjacent to~$v$ whose treedepth is at least that of~$F_u$.
\end{claimproof}

Let~$G'$ and~$F'$ denote the graph and treedepth decomposition upon termination. The following claim proves~\ref{dec:connected:subgraphs}.

\begin{claim} \label{claim:subtrees:induce:connected}
For every~$u \in V(F') \setminus Y$, the graph~$G'[F'_u]$ is connected.
\end{claim}
\begin{claimproof}
Let~$G^2$ and~$F^2$ denote the status of the graph and decomposition after Line~\ref{line:decompose:computef}. By definition of a nice treedepth decomposition, for all~$u \in V(F^2) \setminus Y$, the graph~$G^2[F^2_u]$ is connected. To see that this still holds once the \textbf{while}-loop of Line~\ref{line:while:remove:subtree} has removed parts of the decomposition and the graph, it suffices to observe the following. For every node~$u \in V(F') \setminus Y$ that has survived, no removals were made in the subtree~$F^2_u$: if any removal would have been made, then since we remove entire subtrees rooted at topmost vertices in~$\T$, vertex~$u$ itself would have been removed.
\end{claimproof}

Let~$\T' := \{u \in F' - Y \mid u \mbox{ is a root or } \pi(u) \in Y\}$ as in the lemma statement. We now establish~\ref{dec:component:structure}.

\begin{claim} \label{claim:components:bijection:t}
The vertex sets of the connected components of~$G' - (S \cup Y)$ are exactly the vertex sets of the subtrees of~$F'$ rooted at members of~$\T'$.
\end{claim}
\begin{claimproof}
In one direction, let~$u \in \T'$ and consider the subtree~$F_u$ of~$F$ rooted at~$u$. No descendant of~$u$ is contained in~$Y$, otherwise~$u$ itself would have been included in~$Y_1$ and therefore in~$Y$. By Claim~\ref{claim:subtrees:induce:connected} the graph~$G'[F'_u]$ is connected. Assume that the set~$F'_u$ has a neighbor~$x$ in~$G$ that does not belong to~$S$. Since~$F'$ is a valid treedepth decomposition of~$G' - S$, vertex~$x$ is an ancestor or descendant of a member of~$F'_u$. Since all descendants of~$F'_u$ are contained in~$F'_u$, it follows that~$x$ is a proper ancestor of~$u$. But by definition of~$\T'$, either~$u$ is a root or~$\pi(u) \in Y$, implying that all proper ancestors of~$u$ are in~$Y$. So all vertices in~$N_{G'}(F'_u)$ belong to~$S$ or to~$Y$, proving that each member of~$\T'$ yields a connected component of~$G' - (S \cup Y)$.

For the reverse direction, consider some connected component~$C$ of~$G' - (S \cup Y)$. By Observation~\ref{observation:connectedsubgraph:ancestorpath}, all vertices of~$C$ belong to one tree~$T'$ of~$F'$. Consider the least common ancestor~$u$ of the vertices in~$C$ in tree~$T'$. If~$u \in \T'$ then we are done, since by Claim~\ref{claim:subtrees:induce:connected} the graph~$G'[T'_u]$ is connected and is disjoint from~$S$ and~$Y$; therefore~$C$ must equal~$G'[T'_u]$. Assume for a contradiction that~$u \not \in \T'$. 

If~$u \not \in Y$, then since~$u \not \in \T'$ it follows that~$u$ is not the root of~$T$ and the parent of~$u$ does not belong to~$Y$. But by Claim~\ref{claim:subtrees:induce:connected} the graph~$G'[T'_{\pi(u)}]$ is connected. It is disjoint from~$S$ and disjoint from~$Y$, since all ancestors of~$Y$ are in~$Y$. Hence~$C$ is not a connected component of~$G' - (S \cup Y)$ because there is a connected supergraph of~$C$ in~$G' - (S \cup Y)$.

Finally, consider the case that~$u \in Y$. Since~$u \not \in C$ is the least common ancestor of vertices of~$C$, at least two different children~$c_1,c_2$ of~$u$ contain members~$x_1,x_2$ of~$C$. But by Observation~\ref{observation:connectedsubgraph:ancestorpath}, a common ancestor of~$x_1$ and~$x_2$ is contained in~$H$. But then this must be an ancestor of~$u$. However, all ancestors of~$u$ (including~$u$ itself) are contained in~$Y$, proving that~$C$ intersects~$Y$ and is not a connected component of~$G' - (S \cup Y)$.
\end{claimproof}

The following claim proves~\ref{dec:component:clique}.

\begin{claim}
For every connected component~$C$ of~$G' - (S \cup Y)$, the set~$N_{G'}(C) \cap S$ is a clique in~$G'$.
\end{claim}
\begin{claimproof}
Let~$C$ be a connected component of~$G' - (S \cup Y)$ and assume for a contradiction that~$N_{G'}(C) \cap S$ is not a clique. Let~$\{p,q\} \in N_{G'}(C) \cap S$ be non-adjacent in~$G'$. Then we added a $pq$-separator~$Y_{p,q}$ disjoint from~$p$ and~$q$ to the set~$Y_0$, and it was even a separator in the supergraph of~$G'$ that we considered during Line~\ref{line:add:separator}. Consequently, no connected component of~$G' - Y_0$ can be simultaneously adjacent to both~$p$ and~$q$. Since~$Y \supseteq Y_{p,q}$, it follows that no connected component of~$G' - Y$ can be adjacent to both~$p$ and~$q$; a contradiction.
\end{claimproof}

Finally, we bound the number of connected components of~$G' - (S \cup Y)$ to establish~\ref{dec:component:number}. By Claim~\ref{claim:components:bijection:t} it suffices to bound~$|\T'|$. We partition~$\T'$ into two sets. Let~$\T'_S$ contain the nodes~$u \in \T'$ such that~$N_{G'}(F'_u)$ is a clique; we call these the \emph{simplicial components}. Let~$\T'_N$ be the remaining nodes in~$\T'$, corresponding to \emph{non-simplicial components}.

\begin{claim} \label{claim:simplicial:comp}
$|\T'_S| \leq (|S| + |Y|) (\eta + k)$.
\end{claim}
\begin{claimproof}
Consider a node~$u_0 \in \T'_S$. Since~$N_{G'}(F'_{u_0})$ is a clique, it satisfies the first requirement of the \textbf{while}-loop in Line~\ref{line:decompose:while:removecomponents}. Hence if it was not removed by the algorithm, the second requirement cannot be met. So there is some~$v \in N_{G'}(F_{u_0})$ for which there are no~$\eta + k$ other nodes~$u_i$ in~$\T$ with~$v \in N_{G'}(F_{u_i})$ and~$\td(G'[F'_{u_i}]) \geq \td(G'[F'_{u_0}])$. Charge~$u_0$ to such a neighbor~$v$.

Since~$F'$ is a treedepth decomposition of~$G' - S$, all neighbors of~$F'_{u_0}$ in~$G$ are either contained in~$S$ or are proper ancestors of~$u_0$. Hence all neighbors of~$F'_{u_0}$ are contained in~$S \cup Y$. Now assume for a contradiction that we charge more than~$\eta + k$ nodes of~$\T'_S$ to the same member~$x$ of~$S \cup Y$. Letting~$u_0$ be a node charged to~$x$ that minimizes~$\td(G'[F'_{u_0}])$, we now find that the other~$\eta + k$ nodes charged to~$x$ also have subtrees adjacent to~$x$ that have treedepth at least that of~$G'[F'_{u_0}]$; but then~$u_0$ cannot be charged to~$x$. It follows that we charge at most~$k + \eta$ times to each member of~$S \cup Y$, proving the size bound.
\end{claimproof}

Finally, we bound the number of non-simplicial components.

\begin{claim} \label{claim:nonsimplicial:comp}
$|\T'_N| \leq (|S|^2 + |S| \cdot |Y| + \eta \cdot |Y|)\cdot (\eta + k)$.
\end{claim}
\begin{claimproof}
Consider some~$u \in \T'_N$. By definition of the non-simplicial nodes, there is a pair of vertices~$\{p,q\} \subseteq N_{G'}(F'_u)$ that is not adjacent in~$G$. As observed above, all vertices in~$N_{G'}(F'_u)$ are members of~$S$ or proper ancestors of~$u$ in~$F'$, and are therefore contained in~$Y$. Note that the connected subgraph~$F'_u$ contains the interior vertices of a path between~$p$ and~$q$. By Claim~\ref{claim:distinct:tops:distinct:subtrees}, these paths are pairwise internally vertex-disjoint for different members of~$\T'_N$. Charge every~$u \in \T'_N$ to a pair of non-adjacent vertices in~$N_{G'}(F'_u)$. Since the \textbf{while}-loop of Line~\ref{line:while:add:edge} adds edges between pairs that are connected by~$\eta + k$ pairwise internally vertex-disjoint paths, we can charge at most~$\eta + k$ times to each pair. To prove the claim, it suffices to bound the number of possible pairs. Now observe that every pair~$\{p,q\}$ to which we charge consists of vertices of~$S \cup Y$. The number of pairs where both ends are from~$S$, or exactly one end is from~$S$, is clearly at most~$|S|^2$ and~$|S| \cdot |Y|$, respectively. Finally, observe that for pairs where both members are from~$Y$, these members are in ancestor-descendant relation in~$F'$ since both endpoints are ancestors of the nodes~$u$ that charge to them. Since the height of~$F'$ is at most~$\eta$, each node in~$F'$ has less than~$\eta$ ancestors. If we thus count, for each node in~$Y$, the number of pairs where the other node is higher in the forest, we count at most~$\eta$ incident pairs per vertex of~$Y$, for a total of at most~$\eta \cdot |Y|$. Hence the total number of pairs to which we charge is at most~$|S|^2 + |S| \cdot |Y| + \eta \cdot |Y|$. As we charge at most~$\eta + k$ nodes of~$\T'_N$ to each pair, the claim follows.
\end{claimproof}

Since~$|\T| = |\T_N| + |\T_S|$, by combining Claims~\ref{claim:simplicial:comp},~\ref{claim:nonsimplicial:comp} and~\ref{claim:components:bijection:t} we establish~\ref{dec:component:number}. This concludes the proof of Lemma~\ref{lemma:decomposition:algorithm}.
\end{proof}
\subsection{Reduction algorithm} \label{section:reduction:algorithm}

\begin{algorithm}[t]
\caption{Reduce(Graph~$G$, treedepth-$\eta$ modulator~$S$, treedepth-$\eta$ decomposition~$F$ of~$G - S$, node~$v$ of~$F$, $k \in \mathbb{N}$)} \label{alg:reduce}
\begin{algorithmic}[1]

\vspace{0.2cm}

\STATE Let~$T$ be the tree in~$F$ containing~$v$

\WHILE{$\exists$ distinct~$p, q \in N_G(T_v) \cup \{v\}$ with~$\{p,q\} \not \in E(G)$ and~$\lambda_{G[\{p,q\} \cup T_v]}(p,q) \geq 3\eta$}\label{line:while:add:edge}
	\STATE Add the edge~$\{p,q\}$ to~$G$\label{line:add:edge}
\ENDWHILE

\WHILE{$\exists$ distinct children~$c_0, c_1, \ldots, c_{3\eta}$ of~$v$ such that~$c_0$ has a neighbor~$s \in S$,~$N_G(T_{c_0}) \subseteq N_G[s]$, and for~$i \in [3\eta]$ we have~$\td(G[T_{c_i}]) \geq \td(G[T_{c_0}])$ and~$s \in N_G(T_{c_i})$}\label{line:while:remove:edge}
	\STATE Remove the edges between~$s$ and members of~$T_{c_0}$ from graph~$G$\label{line:remove:edges}
\ENDWHILE

\WHILE{$\exists$ a child~$c^*$ of~$v$ such that~$N_G(T_{c^*})$ is a clique, and for every~$w \in N_G(T_{c^*})$ there are~$3\eta$ distinct children~$c^w_1, \ldots, c^w_{3\eta} \neq c^*$ of~$v$ such that for all~$i \in [3\eta]$ we have~$\td(G[T_{c^w_i}]) \geq \td(G[T_{c^*}])$ and~$w \in N_G(T_{c^w_i})$}\label{line:while:remove:subtree}
	\STATE Remove the vertices in~$T_{c^*}$ from~$F$ and from~$G$\label{line:remove:subtree}
\ENDWHILE

\FOREACH{remaining child~$c$ of~$v$ in~$T$}\label{line:recurse}
	\STATE Reduce($G$, $S$, $F$, $c$, $k$)
\ENDFOR
\end{algorithmic}
\end{algorithm}

\noindent The reduction algorithm that will be applied to each piece of the decomposition is given as Algorithm~\ref{alg:reduce}. To prove that it works correctly, we will prove that it maintains a set of concrete invariants.

\begin{definition}[Invariants] Consider an execution of Reduce($G, S, F, v, k)$. Let~$T$ be the tree of~$F$ containing~$v$ and let~$G^0,F^0,T^0$ be the status of~$G,F$ and~$T$ at the start of the iteration. We define the following invariants of Algorithm~\ref{alg:reduce}.
\begin{enumerate}[(I)]
	\item $F$ is a treedepth decomposition of~$G - S$ of height at most~$\eta$.\label{invar:treedec}
	\item For every vertex~$u \in T_v$ the graph~$G[T_u]$ is connected.\label{invar:nice}
	\item The set~$N_G(T_v) \cap S = N_G(T_v) \setminus \anc_T(v)$ is a clique in~$G$.\label{invar:almostclique}
	\item The graph~$G$ can be obtained from~$G^0$ by
	\begin{itemize}
		\item adding edges whose endpoints belong to~$N_{G^0}(T^0_v) \cup \{v\}$, 
		\item removing edges between~$N_{G^0}(T^0_v) \cap S$ and proper descendants of~$v$ in~$T^0$, 
		\item removing vertex sets of subtrees rooted at children of~$v$.
	\end{itemize}
	\item $F$ is a rooted subforest of~$F^0$.
	\item For every~$u \in T_v$ we have~$N_G(T_u) \cap S \subseteq N_{G^0}(T^0_u)$.\label{invar:neighbors:t}
	\item\label{invar:equivalent} The instance~$(G,k)$ is equivalent to the instance~$(G^0, k)$.
\end{enumerate}
\end{definition}

\begin{lemma} \label{lemma:invariants}
The Reduce algorithm preserves its invariants.
\end{lemma}
\begin{proof}
We will prove that if the invariants hold, then any step taken by the algorithm preserves the invariants. For concreteness, we denote by~$G^0$ and~$F^0$ the state of~$G$ and~$F$ at the time the procedure is called. During the execution of the algorithm, the structures~$G$ and~$F$ change. Let~$T^0$ be the tree in~$F^0$ containing~$v$. The proof is by induction on the height of~$T^0_v$, which is at least one. Assume that the invariants hold before some step of the algorithm and let~$G, F, T$ denote the status of the structures before the step. We use~$G', F', T'$ for the status after the step. We make a distinction based on the action taken by the algorithm.

\textbf{Adding an edge.} Suppose that the algorithm adds an edge~$\{p,q\}$ in Line~\ref{line:add:edge} so that~$G' := G + \{p,q\}$. To see that~$F$ is still a valid treedepth decomposition of~$G' - S$, it suffices to observe that the added edge either has an endpoint in~$S$, or both its endpoints are ancestors of~$v$, implying that the edge is represented by the decomposition. Hence Invariant~\ref{invar:treedec} is preserved. To see that invariant~\ref{invar:neighbors:t} is preserved, note that both endpoints of the added edge are contained in~$N_G(T_v) \cup \{v\}$, and so the only vertex in~$T_v$ that can be incident on the added edge is~$v$ itself. If an edge was added from~$v$ to a vertex~$s \in S$ then some member of~$T_v$ was already adjacent to~$s$. The only other invariant that is not trivially maintained is Invariant~\ref{invar:equivalent}. To see that it is maintained as well, observe the following.

By Invariant~\ref{invar:almostclique}, the set~$N_{G}(T_v) \setminus \anc_{T}(v)$ is a clique. Since the edge~$\{p,q\}$ we add either has an endpoint in~$\anc_{T}(v)$, or is an edge between~$v$ and a member of~$N_{G}(T_v) \setminus \anc_{T}(v) = N_{G}(T_v) \cap S$, it follows that~$N_{G'}(T'_v) \setminus \anc_{T'}(v)$ is also a clique (the decomposition tree does not change). Hence the set~$T'_v$ is $\eta$-nearly clique separated in~$G'$, since~$v$ has at most~$\eta$ ancestors. By Lemma~\ref{lem:boundingsolnintersection:updated}, any minimum treedepth-$\eta$ modulator of~$G'$ intersects~$T'_v$ in at most~$2\eta$ vertices. We may therefore apply Lemma~\ref{lemma:remove:edges} where the set~$T'_v$ is used as~$X$ and~$\ell = 2\eta$, to establish that~$(G' = G + \{p,q\}, k)$ is equivalent to~$(G, k)$ and therefore, using the invariant applied to~$G$, to~$(G^0,k)$. This proves that Invariant~\ref{invar:equivalent} is maintained. 

\textbf{Removing a set of edges.} Now consider what happens when the algorithm removes a set of edges in Line~\ref{line:remove:edges}. Since the edges we remove have exactly one endpoint in~$S$, all invariants except Invariant~\ref{invar:equivalent} are easily seen to be preserved. To prove that~\ref{invar:equivalent} is also preserved, we will apply Lemma~\ref{lemma:remove:edges}. Let us consider the requirements for the lemma. Define~$\ell := 2\eta$ and let~$X_1, \ldots, X_{\ell + \eta}$ be the vertex sets of~$T_{c_1}, \ldots, T_{c_{3\eta}}$ identified in the algorithm. By Invariant~\ref{invar:nice}, for every~$i \in [3\eta]$ the graph~$G[T_{c_i}]$ is connected. The condition in the \textbf{while} loop ensures that~$\td(G[T_{c_i}]) \geq \td(G[T_{c_0}])$ for all~$i \in [3\eta]$. Let~$S$ be the vertices of~$T_{c_0}$. It follows that our choice of~$S$ and the~$X_i$ satisfy the first two conditions of Lemma~\ref{lemma:remove:edges}, when using the vertex~$s$ in the algorithm as~$v$ in the lemma statement. To see that the third condition is also valid, observe that~$\X := \bigcup _{i \in [3\eta]} X_i$ is contained in~$T_v$ and that in any graph obtained from~$G$ by removing edges between~$s$ and~$T_{c_0}$, the set~$T_v$ is $\eta$-nearly clique separated by Invariant~\ref{invar:almostclique}. Hence, by Lemma~\ref{lem:boundingsolnintersection:updated}, any minimum treedepth modulator in a graph obtained from~$G$ by removing edges between~$s$ and~$T_{c_0}$ contains at most~$2\eta = \ell$ vertices from~$T_v$. Together with the fact that~$N_G(T_{c_0}) \subseteq N_G[s]$ we find that all conditions of Lemma~\ref{lemma:remove:edges} are satisfied, which proves that instance~$(G',k)$ is equivalent to~$(G,k)$. By Invariant~\ref{invar:equivalent} and transitivity, instance~$(G',k)$ is equivalent to~$(G^0,k)$ and therefore said invariant is preserved.

\textbf{Removing the vertices of a child subtree.} As the next operation, suppose that the Reduce algorithm removes the vertices in a subtree rooted at a child~$c^*$ of~$v$, in Line~\ref{line:remove:subtree}. Then~$N_G(T_{c^*})$ is a clique and for every~$w \in N_G(T_{c^*})$ there are~$3\eta$ distinct children~$c^w_1, \ldots, c^w_{3\eta}$ unequal to~$c^*$ such that the treedepth of the subgraphs they represent is at least~$\td(G[T_{c^*}])$, and they each contain a neighbor of~$w$. Observe that by Invariant~\ref{invar:nice}, for any~$i \in [3\eta]$ and~$w \in N_G(T_{c^*})$ the graph~$G[T_{c^w_i}]$ is connected. Set~$\ell := 2\eta$. We will prove that the conditions of Lemma~\ref{lemma:remove:simplicial:component} are satisfied for this choice of~$\ell$, using~$T_{c^*}$ as~$S$. Observe that the height of~$T_{c^*}$ is at most~$\eta$ by Invariant~\ref{invar:treedec}. The previous observations ensure that, when choosing~$X^w_i$ as~$T_{c^w_i}$ for all~$w \in N_G(T_{c^*})$ and~$i \in [3\eta]$, the first condition of Lemma~\ref{lemma:remove:simplicial:component} is satisfied. The fact that for each choice of~$w$, for all~$i$ the defined sets~$X^w_i$ are pairwise disjoint (note that sets for different choices of~$w$ may overlap) follows from the fact that the~$X^w_i$ come from different children of~$v$. Hence the second condition of the lemma is also satisfied. To see that the last condition is satisfied, note that~$N_G(T_v) \setminus \anc_T(v)$ is a clique by Invariant~\ref{invar:almostclique}, and hence~$T_v$ is $\eta$-nearly clique separated in~$G$. Therefore, using Lemma~\ref{lem:boundingsolnintersection:updated}, any minimum treedepth-$\eta$ modulator in~$G$ intersects~$T_v$ in at most~$\ell = 2\eta$ vertices. Hence the final condition of Lemma~\ref{lemma:remove:simplicial:component} is satisfied, proving that~$(G',k)$ is equivalent to~$(G,k)$ and therefore to~$(G^0,k)$. This implies that Invariant~\ref{invar:equivalent} is satisfied.

Let us now consider the other invariants. The only remaining invariant that is not trivially maintained is Invariant~\ref{invar:nice}, which says that the graph~$G[T_u]$ is connected for any vertex~$u$ in~$T_v$. Let~$T'$ denote the status of~$T$ after the deletion of~$T_{c^*}$. Observe that since we removed an entire subtree rooted at a child of~$v$, the only vertex~$u$ in~$T'_v$ for which~$T'_u$ differs from~$T_u$, is vertex~$v$ itself. To see that~$G'[T'_v]$ is connected, observe that if~$T'_v$ is not a single vertex, then by the properties of treedepth decompositions,~$v$ is a cutvertex in~$G[T_v]$ that separates the vertices in~$T_{c^*}$ from the remaining vertices. Hence any simple path that existed between two vertices of~$T_v \setminus T_{c^*}$ in the graph~$G[T_v]$, still exists in~$G[T'_v]$. It follows that Invariant~\ref{invar:nice} is maintained.

\textbf{Executing a recursive call.} The last case is when the operation that the algorithm performs is making a recursive call. This is where we use induction. If the algorithm makes a recursive call, then this is for children of~$v$ and therefore the height of~$T_v$ is larger than one. Since the height of~$T_u$ is smaller than the height of~$T_v$ for all children~$u$ of~$v$, by induction we find that the recursive call maintains the invariants. This concludes the proof of Lemma~\ref{lemma:invariants}.
\end{proof}

Having established the invariants of the algorithm, we know that it preserves the answer to an instance of \TreedepthEtaDeletion. For the purposes of obtaining a kernel, we also need to prove that it achieves a provable size reduction. We do this in the following lemma.

\begin{lemma}\label{lemma:countleaves}
Let Reduce be called for the input~$(G^0,S^0,F^0,v,k)$, and let~$T^0$ be the tree in~$F^0$ containing~$v$. Let~$G',F'$ be the graph and decomposition once the procedure has finished. Let 
$$\phi(u) := (2 \cdot 3 \eta \cdot 2^\eta)^{\height(T^0_v)} \cdot (|N_{G^0}(T^0_v) \cap S| + 1),$$
for any vertex~$u \in T^0_v$. Then the number of leaves in~$F'_v$ is at most~$\phi(v)$. 
\end{lemma}
\begin{proof}
We will prove the lemma by induction on~$\height(T^0_v)$. Before doing so, however, we establish the structure of the graph once the Reduce algorithm has reached Line~\ref{line:recurse}. Let~$G$ be the state of the graph once reaching Line~\ref{line:recurse}, let~$F$ be the state of the forest, and let~$T$ be the tree of~$F$ containing~$v$. 

\begin{claim} \label{claim:paths:from:subtrees}
Let~$u$ be a child of~$v$ in~$T$. Then for every pair of distinct vertices~$\{p,q\} \subseteq N_G(T_u)$ there is a $pq$-path in the graph~$G[\{p,q\} \cup T_u]$.
\end{claim}
\begin{claimproof}
If~$p,q \in N_G(T_u)$ then there is a neighbor~$p'$ of~$p$ in~$T_u$, and a neighbor~$q'$ of~$q$ in~$T_u$. By Invariant~\ref{invar:nice}, the graph~$G[T_u]$ is connected and contains a path~$P$ connecting~$p'$ and~$q'$. By adding the edges to~$p$ and~$q$ we obtain the desired~$pq$ path in~$G[\{p,q\} \cup T_u]$.
\end{claimproof}

We now derive bounds on the number of children of~$v$ once the execution has reached Line~\ref{line:recurse}. Let~$\C^+$ denote the set of children~$u$ of~$v$ for which~$N_G(T_u) \cap S \neq \emptyset$. Let~$\C^-$ denote the remaining children of~$v$.

\begin{claim} \label{claim:bound:cminus}
$|\C^-| \leq 3\eta \cdot 2^{\depth(v,T^0)}$.
\end{claim}
\begin{claimproof}
Consider a child~$u$ of~$v$ such that~$N_G(T_u) \cap S = \emptyset$. By Observation~\ref{observation:neighbors:of:subtree}, the only possible vertices of~$N_G(T_u)$ are proper ancestors of~$u$, which are exactly the ancestors~$Y$ of~$v$ (since~$v$ is its own ancestor). There are exactly~$\depth(v,T) = \depth(v,T^0)$ of them. For~$Y' \subseteq Y$ let~$\C^-_{Y'}$ contain the children~$u$ of~$v$ for which~$N_G(T_u) = Y'$. Since there are~$2^{|Y|} = 2^{\depth(v,T^0)}$ possible groups, to establish the claim it suffices to bound the size of each group by~$3 \eta$.

Fix some~$Y' \subseteq Y$ and the group~$\C^-_{Y'}$. Assume for a contradiction that~$|\C^-_{Y'}| > 3\eta$. For any pair of distinct vertices~$\{p,q\} \in Y'$ there are at least~$3\eta$ internally vertex-disjoint $pq$-paths in the graph~$G[\{p,q\} \cup T_v]$, since we get one such path through each set~$T_{c}$ with~$c \in \C^-_{Y'}$, by Claim~\ref{claim:paths:from:subtrees}. If these paths exist in~$G$, then surely they must have existed in the state of the graph when the \textbf{while}-loop of Line~\ref{line:while:add:edge} terminated, since we only delete vertices and edges afterward. So the paths were detected in that loop, causing the edge~$\{p,q\}$ to be added to~$G$. Since such edges are not removed in the \textbf{while}-loop of Line~\ref{line:while:remove:edge}, during the \textbf{while}-loop of Line~\ref{line:while:remove:subtree} for each~$c \in \C^-_{Y'}$ the set~$N_G(T_c)$ is a clique. But if~$|\C^-_{Y'}| > 3\eta$, then letting~$c^*$ be a member of~$\C^-_{Y'}$ minimizing~$\td(G[T_{c^*}])$ and letting~$c_1, \ldots, c_{3\eta}$ be~$3\eta$ arbitrary other members of~$\C^-_{Y'}$, since all subtrees rooted at~$\C^-_{Y'}$ have the same $G$-neighborhood we now find that this choice of~$c^*$ and using the same~$c_1, \ldots, c_{3\eta}$ for all~$w \in N_G(T_c)$, the conditions of the \textbf{while}-loop are satisfied, causing~$c^*$ to be deleted. Hence if the algorithm was correctly executed,~$|\C^-_{Y'}| \leq 3\eta$. The claim follows.
\end{claimproof}

\begin{claim} \label{claim:sum:s:degrees}
$\sum _{u \in \C^+} |N_G(T_u) \cap S| \leq 3 \eta \cdot 2^{\depth(v,T^0)} \cdot |N_{G^0}(T^0_v) \cap S|.$
\end{claim}
\begin{claimproof}
For each child~$u$ of~$v$, by Observation~\ref{observation:neighbors:of:subtree} we know that~$N_G(T_u) \setminus S$ consists of ancestors of~$v$. As in the proof of the previous claim, let~$Y$ be the ancestors of~$v$ and partition the set~$\C^+$ into subsets~$\C^+_{Y'}$ for~$Y' \subseteq Y$ such that for all~$u \in \C^+_{Y'}$ we have~$N_G(T_u) \setminus S = Y'$.

Fix an arbitrary~$Y' \subseteq Y$. We start by showing that for every vertex~$s$ of~$S$, there are at most~$3\eta$ members~$u$ of~$\C^+_{Y'}$ such that~$s \in N_G(T_u)$. Assume for a contradiction that some~$s \in S$ is adjacent to more than~$3\eta$ of the subtrees rooted at~$\C^+_{Y'}$. For each subtree~$u \in \C^+_{Y'}$, by Claim~\ref{claim:paths:from:subtrees} for any vertex~$y \in Y'$ there is an $su$-path in subgraph~$G[\{s,y\} \cup T_u]$. Since all subtrees~$T_u$ are contained in~$T_v$, this implies that~$\lambda_{G[\{s,y\} \cup T_v]}(s,y) \geq 3\eta$ for every~$y \in Y'$. Hence the edge addition rule triggered in Line~\ref{line:add:edge} for~$s$ and every member of~$Y'$, showing that~$s$ is adjacent to all members of~$Y'$. Now let~$c_0 \in \C^+_{Y'}$ minimize~$\td(G[T_{c_0}])$ among all members of~$\C^+_{Y'}$ for which~$s \in N_G(T_{c_0})$, and let~$c_1, \ldots, c_{3\eta}$ be~$3\eta$ members~$u$ of~$\C^+_{Y'}$ for which~$s \in N_G(T_u)$. Then~$s \in N_G(T_{c_0})$ and for all~$i \in [3\eta]$ we have~$s \in N_G(T_{c_i})$. By choice of~$c_0$ we know~$\td(G[T_{c_i}]) \geq \td(G[T_{c_0}])$ for all~$i \in [3\eta]$. Finally, since~$s$ is adjacent to all members of~$Y'$ and~$N_G(T_{c_0}) \setminus Y' \subseteq S$ is a clique containing~$s$, it follows that~$N_G(T_{c_0}) \subseteq N_G[s]$. But then all conditions of the \textbf{while}-loop of Line~\ref{line:while:remove:edge} are satisfied, showing that the algorithm would have removed the edges from~$s$ to~$T_{c_0}$; a contradiction. Hence we established that for every~$s \in S$ there are at most~$3\eta$ members~$u$ in~$\C^+_{Y'}$ with~$s \in N_G(T_u)$.
Let us finish the proof of the claim with this knowledge. Observe the following double-counting equality:
$$\sum _{u \in \C^+_{Y'}} |N_G(T_u) \cap S| = \sum _{s \in S} \left|\{u \in \C^+_{Y'} \mid s \in N_G(T_u)\}\right|.$$
By the previous argument, every~$s \in S$ is a neighbor of at most~$3\eta$ subtrees rooted in~$\C^+_{Y'}$. In addition, observe that every~$s \in S$ that has a neighbor in~$T_u$ for some~$u \in \C^+_{Y'}$ trivially also has a neighbor in~$T_v \supset T_u$. Hence vertices~$s \in S \setminus N_G(T_v)$ have no neighbors among such subtrees~$T_u$. This proves that:
$$\sum _{s \in S} \left |\{u \in \C^+_{Y'} \mid s \in N_G(T_u)\} \right | = \sum _{s \in N_G(T_v) \cap S} \left |\{u \in \C^+_{Y'} \mid s \in N_G(T_u)\} \right | \leq |N_G(T_v) \cap S| \cdot 3\eta.$$
Combining these two inequalities, and summing over all possible~$Y'$, we find that:
\begin{align*}
\sum _{u \in \C^+} |N_G(T_u) \cap S| & \leq \sum _{Y' \subseteq Y} \sum _{u \in \C^+_{Y'}} |N_G(T_u) \cap S| \\
&= \sum _{Y' \subseteq Y} \sum _{s \in N_G(T_v) \cap S} \left|\{u \in \C^+_{Y'} \mid s \in N_G(T_u)\} \right| \\
&\leq 2^{|Y|} \cdot |N_G(T_v) \cap S| \cdot 3\eta
\end{align*}

Observe that~$|Y| = \depth(v,T^0)$. To establish the claim, observe that by Invariant~\ref{invar:neighbors:t} we have~$|N_G(T_v) \cap S| \leq |N_{G^0}(T^0_v) \cap S|$.
\end{claimproof}

Using the previous claims we now bound the number of leaves that remain in the subtree~$T_v$ after the algorithm has terminated. Observe that once the algorithm recurses on a child~$u$, the height of~$T_u$ is less than the height of~$T_v$ and therefore we may apply induction to bound the number of leaves in the subtrees resulting from recursive calls.

If there are no children of~$v$ left to recurse on, then~$v$ is the only leaf in~$T_v$ that remains. It is easy to see that the bound of Lemma~\ref{lemma:countleaves} ensures that~$\phi(u) \geq 1$ and therefore the base case of the induction holds. If~$v$ has at least one child left, then~$v$ is not a leaf. The number of leaves in~$T_v$ is obtained by summing the bounds for its children. Let~$f(v)$ denote the number of leaves of the subtree~$T_v$ after the procedure terminated and observe the following derivation. We count the leaves in children rooted at~$\C^+$ and the leaves in children rooted at~$\C^-$ separately.
\begin{align*}
& \sum _{u \in \C^-} f(u) \leq \sum _{u \in \C^-} \phi(u) & \mbox{By induction.} \\
& \leq \sum _{u \in \C^-} (3 \eta \cdot 2^{\eta + 1})^{\height(T^0_u)} \cdot (|N_{G^0}(T^0_u) \cap S| + 1) \\
& \leq \sum _{u \in \C^-} (3 \eta \cdot 2^{\eta + 1})^{\height(T^0_u)} \cdot (0 + 1) &  \mbox{Definition of~$\C^-$.}\\
& \leq \sum _{u \in \C^-} (3 \eta \cdot 2^{\eta + 1})^{\height(T^0_v) - 1} & \mbox{$u$ is below~$v$.}\\
& \leq |\C^-| \cdot  (3 \eta \cdot 2^{\eta + 1})^{\height(T^0_v) - 1} \\
& \leq 3\eta \cdot 2^{\depth(v,T^0)} \cdot (3 \eta \cdot 2^{\eta + 1})^{\height(T^0_v) - 1} & \mbox{Claim~\ref{claim:bound:cminus}.} \\
&\leq (3 \eta \cdot 2^{\eta + 1})^{\height(T^0_v)} & \mbox{$\depth(v,T^0) \leq \eta$}.
\end{align*}
Now we consider~$\C^+$.
\begin{align*}
& \sum _{u \in \C^+} f(u) \leq \sum _{u \in \C^+} \phi(u) & \mbox{By induction.} \\
& \leq \sum _{u \in \C^+} (3 \eta \cdot 2^{\eta + 1})^{\height(T^0_u)} \cdot (|N_{G^0}(T^0_u) \cap S| + 1) \\
& \leq \sum _{u \in \C^+} (3 \eta \cdot 2^{\eta + 1})^{\height(T^0_v) - 1} \cdot (|N_{G^0}(T^0_u) \cap S| + 1) & \mbox{$u$ is below~$v$.}\\
& \leq (3 \eta \cdot 2^{\eta + 1})^{\height(T^0_v) - 1} \cdot \sum _{u \in \C^+} (|N_{G^0}(T^0_u) \cap S| + 1) & \mbox{Rearranging.}\\
& \leq (3 \eta \cdot 2^{\eta + 1})^{\height(T^0_v) - 1} \cdot 2 \sum _{u \in \C^+} |N_{G^0}(T^0_u) \cap S| & \mbox{$N_{G^0}(T^0_u) \cap S| > 0$ by def.~$\C^+$.}\\
& \leq (3 \eta \cdot 2^{\eta + 1})^{\height(T^0_v) - 1} \cdot 2 \cdot 3 \eta \cdot 2^{\depth(v,T^0)} \cdot |N_{G^0}(T^0_v) \cap S| & \mbox{Claim~\ref{claim:sum:s:degrees}.}\\
& \leq (3 \eta \cdot 2^{\eta + 1})^{\height(T^0_v)} \cdot |N_{G^0}(T^0_v) \cap S| & \mbox{$\depth(v,T^0) \leq \eta$.}
\end{align*}
Since~$f(v) = (\sum _{u \in \C^+}f(u)) + (\sum_{u \in \C^-}f(u))$, combining the bounds for~$\C^+$ and~$\C^-$ to bound~$f(v)$ proves Lemma~\ref{lemma:countleaves} by simple formula manipulation.
\end{proof}

Finally, since a kernelization is an \emph{efficient} preprocessing algorithm, we have to show that the Reduce algorithm can be implemented efficiently.

\begin{lemma}
The Reduce algorithm can be implemented to run in polynomial time for every fixed~$\eta$.
\end{lemma}
\begin{proof}
The nontrivial algorithm tasks that have to be done for one iteration of the algorithm are determining the treedepth of the subgraph induced by a subtree of~$F$, and finding the maximum number of internally vertex-disjoint paths in some subgraph. For every fixed~$\eta$, the first can be done in polynomial (even linear) time, using for example the FPT algorithm of Reidl et al.~\cite{ReidlRVS14}. It is well known that the number of internally vertex-disjoint paths can be computed in polynomial time using flow techniques; see, for example, Schrijver~\cite[Chapter 9]{Schrijver03}. Since the Reduce algorithm recurses at most once on each child, it follows that the overall runtime is polynomial in the size of the input graph.
\end{proof}

\subsection{Final kernelization algorithm} \label{section:final:kernel}

Armed with the decomposition and the reduction algorithm, we can formulate the complete kernelization algorithm for \TreedepthEtaDeletion.

\begin{untheorem}
For each fixed~$\eta$, \TreedepthEtaDeletion has a polynomial kernel with~$\Oh(k^6)$ vertices: an instance~$(G,k)$ can be efficiently reduced to an equivalent instance~$(G',k)$ with~$2^{\Oh(\eta^2)} k^6$ vertices.
\end{untheorem}
\begin{proof}
Fix some~$\eta \geq 1$. When presented with an input~$(G, k)$, the kernelization algorithm proceeds as follows. It first applies Lemma~\ref{lemma:decomposition:algorithm}. If the lemma reports that~$G$ has no treedepth-$\eta$ modulator of size at most~$k$ then we output a constant-size \no-instance and terminate. Otherwise we obtain an equivalent instance~$(G',k)$ along with a treedepth-$\eta$ modulator~$S$, a set~$Y \subseteq V(G') \setminus S$, a treedepth decomposition~$F'$ of~$G'$ of height at most~$\eta$ satisfying the conditions outlined in Lemma~\ref{lemma:decomposition:algorithm}. The graph~$G'$ consists of~$S$,~$Y$, and the vertices of~$G' - (S \cup Y)$. The lemma guarantees that the vertex sets of connected components of~$G' - (S \cup Y)$ correspond to the vertex sets of subtrees of~$F'$ rooted at members of~$\T'$. It also gives a bound on the cardinality of~$|\T'|$, by bounding the number of connected components of~$G' - (S \cup Y)$. Since~$|S|$ and~$|Y|$ are already small, to bound the total size of the instance it suffices to shrink each component of~$G' - (S \cup Y)$ to size polynomial in~$k$. 

To achieve this, we do the following. For each~$v \in \T'$ we call Reduce$(G', S, F', v, k)$. From the guarantees of Lemma~\ref{lemma:decomposition:algorithm}, it follows that the invariants outlined in Section~\ref{section:reduction:algorithm} are satisfied for the first call of Reduce. Since the changes that are made to the graph by the reduction algorithm are local, as formalized in the invariant, when we call Reduce on the next member of~$\T'$ the invariants are still initially satisfied. As Invariant~\ref{invar:equivalent} ensures that each transformation yields an instance equivalent to the one we started with, after executing Reduce for each~$v \in \T'$ the resulting instance~$(G',k)$ is equivalent to the original input. Since Decompose and Reduce both run in polynomial time for fixed~$\eta$, the entire procedure runs in polynomial time. The resulting instance~$(G',k)$ is given as the output of the kernelization. It remains to bound the number of vertices in~$G'$.

The vertex set of the final graph~$G'$ consists of~$S$,~$Y$, and whatever is left of the subtrees rooted at~$\T'$. By Lemma~\ref{lemma:countleaves}, after Reduce has finished processing for~$v \in \T'$, the number of leaves in the subtree of~$F'$ rooted at~$v$ has been reduced to at most~$\phi(v) = (3\eta \cdot 2^{\eta+1})^{\height(T_v)} \cdot (|S| + 1)$. As the height of any tree in~$F'$ is at most~$\eta$, any leaf has at most~$(\eta - 1)$ proper ancestors. Hence the number of vertices in the reduced subtree~$F'_v$ is at most~$\eta$ times the number of leaves, so at most~$$\eta \cdot (3\eta \cdot 2^{\eta+1})^{\eta} \cdot (|S| + 1).$$ By combining this bound with the number of connected components of~$G' - (S \cup Y)$, which equals~$|\T'|$ by Lemma~\ref{lemma:decomposition:algorithm}, we can now obtain a final size bound for~$G'$. Observe that, for fixed~$\eta$, Lemma~\ref{lemma:decomposition:algorithm} shows that~$|S| \in \Oh(k)$ and~$|Y| \in \Oh(k^3)$, implying that~$|\T'| \in \Oh(k^5)$. Hence we find:

\begin{align*}
|V(G')| &\leq |S| + |Y| + |\T'| \cdot \eta \cdot (3\eta \cdot 2^{\eta+1})^{\eta} \cdot (|S| + 1) \\
&\in \Oh(k + k^3 + k^5 \cdot k) \in \Oh(k^6).
\end{align*}

A straight-forward computation shows that the number of vertices in the kernel is bounded by~$2^{c \cdot \eta^2} \cdot k^6$ for an explicit, small constant~$c$ that can be extracted from our arguments. We omit the computations for ease of presentation. This concludes the proof.
\end{proof}
\section{Conclusion}\label{sec:conclusion}
In this paper we (re-)studied the \PlanarFDeletion problem from the perspective of (uniform) kernelization. We answered the question whether all \PlanarFDeletion problems have uniformly polynomial kernels negatively, but showed that the special case \TreedepthEtaDeletion (which is a \PlanarFDeletion problem for every~$\eta$, where every~$\F$ contains a path) has uniformly polynomial kernels.

In a recent paper, Fellows and Jansen~\cite{FellowsJ14} analyzed the connection between kernelization algorithms and minor order obstruction sets, suggesting that the sizes of kernels and the sizes of the graphs in related obstruction sets are closely linked. Our results in this paper give another example of this connection. For every~$k$ and~$d$, the graphs in which at most~$k$ vertices can be deleted to obtain a graph of treewidth at most~$d$ are a minor-closed family, and are therefore characterized by a finite construction set~$\Oh_k$. How do the members of this obstruction set relate to kernelization? We proved that the number of vertices in kernels for \TreewidthDminoneDeletion must be~$\Omega(k^{\frac{d}{4} - \epsilon})$ unless \containment. By a simple construction (although different than that of Lemma~\ref{lemma:lowerbound:constructgraph}), we can prove that there are minor-minimal obstructions that have~$\Omega(k^d)$ vertices. For upper bounds, consider the obstruction set~$\Oh'_k$ for the class of graphs whose treedepth can be reduced to~$\eta$ by at most~$k$ deletions. By adapting arguments from the kernelization---which is needed to circumvent the need to add edges to the graph---we can prove an upper bound of~$2^{\Oh(\eta^2)} k^6$ on the number of vertices of members of the obstruction set~$\Oh'_k$. Hence in the positive case we also see the connection between kernel sizes and obstruction sizes.

The distinction between uniformly versus non-uniformly polynomial kernels is similar to the distinction between algorithms whose parameter dependence is fixed-parameter tractable (\fpt) versus slicewise-polynomial ({\sf XP}), and opens up a similarly broad area of investigation. The kernelization complexity of \FDeletion is still wide open. Some notable open problems in this direction are:
\begin{itemize}
	\item Does \FDeletion admit a polynomial kernel for any fixed set~$\F$, even when~$\F$ contains no planar graphs? Even for the special case of deleting~$k$ vertices to get a planar graph ({\sc Vertex Planarization}), we do not know the answer.
	\item Consider the graphs that can be made planar by at most~$k$ vertex deletions, and the corresponding obstruction set~$\Oh^*_k$ for this family. Can the size of the members of~$\Oh^*_k$ be bounded polynomially in~$k$? By the suggested connection between kernel sizes and obstruction sizes, this may shed light on the kernelization complexity of~\textsc{Vertex Planarization}.
	\item Is it possible to obtain a dichotomy theorem, characterizing the families~$\F$ for which \PlanarFDeletion admits uniformly polynomial kernels?
\end{itemize}

These questions are part of a large research program into the complexity of \FDeletion problems, whose importance was recognized by its listing in the \emph{Research Horizons} section of the recent textbook by Downey and Fellows~\cite[Chapter 33.2]{DowneyF13}.

\newpage
\bibliographystyle{siam}
\bibliography{Paper,book_kernels_fvf,kernels}

\end{document}